\documentclass[12pt]{amsart}
\usepackage{amsmath}
\usepackage{tikz}
\usetikzlibrary{matrix,arrows}
\usepackage{amsfonts,amssymb,amsthm, txfonts, pxfonts,amscd}
\usepackage[stretch=10]{microtype}
\def\struckint{\mathop{%
\def\mathpalette##1##2{\mathchoice{##1\displaystyle##2}%
 {##1\textstyle##2}{##1\scriptstyle##2}{##1\scriptscriptstyle##2}}%
\mathpalette
{\vbox\bgroup\baselineskip0pt\lineskiplimit-1000pt\lineskip-1000pt
\halign\bgroup\hfill$}
{##$\hfill\cr{\intop}\cr\diagup\cr\egroup\egroup}%
}\limits}
\usepackage{natbib}
\usepackage{color}
\usepackage{booktabs,caption,fixltx2e}
\usepackage[flushleft]{threeparttable}
\usepackage{amsmath}
\usepackage{tikz}
\usetikzlibrary{matrix,arrows}
\usepackage{amsfonts,amssymb,amsthm, txfonts, pxfonts,amscd}
\usepackage{multirow}
\usepackage{floatrow}
\usepackage{subcaption}
\newfloatcommand{capbtabbox}{table}[][\FBwidth]

\usepackage[margin=1in]{geometry}
\usepackage{setspace}
\usepackage{soul,color, mathtools, fixmath}
\usepackage{multirow}
\usepackage[shortlabels]{enumitem}
\def\pr{\mathop{\text{pr}}\nolimits}
\def\k{{\bf k}}
\def\z{{\bf z}}
\def\W{{\bf W}}

\def\pr{\mathop{\text{pr}}\nolimits}
\def\E{\mathcal{E}}

\newtheorem{thm}{Theorem}[section]
\newtheorem{lemma}[thm]{Lemma}

\newtheorem{assumption}[thm]{Assumption}
\newtheorem{rmk}[thm]{Remark}
\makeatletter

\def\indep{\mathrel{\rlap{$\perp$}\kern1.6pt\mathord{\perp}}}

\def\H{\mathcal{H}}
\def\E{\mathbb{E}}
\def\Y{{\bf Y}}

\def\given{\, | \,}
\def\Given{\, \big | \,}

\def\bft{{\bf t}}
\def\bfx{{\bf x}}

\def\bfT{{\bf T}}
\def\bfD{{\bf D}}
\def\dotminussym#1#2{%
  \setbox0=\hbox{$\m@th#1-$}%
  \kern.5\wd0%
  \hbox to 0pt{\hss\hbox{$\m@th#1-$}\hss}%
  \raise.6\ht0\hbox to 0pt{\hss$\m@th#1.$\hss}%
  \kern.5\wd0}

\mathchardef\mhyphen="2D

\begin{document}

\title[Recurrent event analysis with functional covariates via random subsampling]{Recurrent event analysis in the presence of real-time high frequency data via random subsampling}
\author{Walter Dempsey}
\address {Department of Biostatistics, University of Michigan, 1415 Washington Heights, Ann Arbor, MI 48109, USA}
 \email{wdem@umich.edu}

\date{\today}

\begin{abstract}
Digital monitoring studies collect real-time high frequency data via mobile sensors in the subjects' natural environment. This data can be used to model the impact of changes in physiology on recurrent event outcomes such as smoking, drug use, alcohol use, or self-identified moments of suicide ideation. Likelihood calculations for the recurrent event analysis, however, become computationally prohibitive in this setting. Motivated by this, a random subsampling framework is proposed for computationally efficient, approximate likelihood-based estimation. A subsampling-unbiased estimator for the derivative of the cumulative hazard enters into an approximation of log-likelihood. The estimator has two sources of variation: the first due to the recurrent event model and the second due to subsampling. The latter can be reduced by increasing the sampling rate; however, this leads to increased computational costs.  The approximate score equations are equivalent to logistic regression score equations, allowing for standard, ``off-the-shelf'' software to be used in fitting these models. Simulations demonstrate the method and efficiency-computation trade-off. We end by illustrating our approach using data from a digital monitoring study of suicidal ideation.
\end{abstract}


\maketitle

\section{Introduction}

Advancement in mobile technology has led to the rapid integration of mobile and wearable sensors into behavioral health~\citep{Freeetal2013}. Take HeartSteps, for example, a mobile health (mHealth) study designed to increase physical activity in sedentary adults \citep{KlasnjaHS2019}. Here, a Jawbone sensor is used to monitor step count every minute of the participant's study day. Of interest in many mHealth studies is the relation of such real-time high frequency sensor data to an adverse, recurrent event process. In a smoking cessation mHealth study~\citep{Sense2Stop}, for example, the relation between a time-varying sensor-based measure of physiological stress and smoking lapse is of scientific interest. In a suicidal ideation mHealth study~\citep{Kleiman2018}, the relation of electrodermal activity (EDA) and accelerometer with self-identified moments of suicidal ideation is of scientific interest.

The goal of this paper is two-fold: (1) to discuss the appropriate choice of statistical model for joint high-frequency sensor data and the recurrent event data, and (2) to construct a simple, easy-to-implement method for parameter estimation and inference. For (1), we discuss an important issue regarding measurement-error models when paired with recurrent event outcomes. For (2), we introduce a random subsampling procedure that has several benefits.  First, the resulting inference is unbiased; however, there is a computation-efficiency trade-off. In particular, a higher sampling rate can decrease estimator variance at the cost of increased computation.  We show via simulations that the benefits of incredibly high sampling rates is often negligible, as the contribution to the variation is small relative to the variation in the underlying stochastic processes. Second, derived estimating equations are optimal, implying loss of statistical efficiency is only due to subsampling procedure and not the derived methodology.  Finally, implementation can leverage existing, standard software for functional data analysis and logistic regression, leading to fast adoption by domain scientists.

\section{Recurrent event process and associated high frequency data}

Suppose $n$ subjects are independently sampled with observed event times~$\bfT_{i} = \{ T_{i,1}, \ldots, T_{i,k_i}\}$ over some observation window $[0, \tau_i]$ for each subject $i = 1,\ldots, n$.  Assume the event times are ordered, i.e., $T_{i,j} < T_{i,j^\prime}$ for $j < j^\prime$. The observation window length, $\tau_i$, is the censoring time and is assumed independent of the event process. Let~$N_{i} (t)$ denote the associated counting process of $\bfT_{i}$; that is, $N_i (t) = \sum_{j=1}^{k_i} 1 [ T_{i,j} < t ]$. In this section, we assume a single-dimensional health process~$\bfx_i = \{ x_i (s) \}_{0 < s < \tau_i}$ for each participant is measured at a dense grid of time points.  Accelerometer, for example, is measured at a rate of 32Hz (i.e., $32$ times per second). Electrodermal activity (EDA), on the other hand, is measured at a rate of 4Hz (i.e., 4 times per second).  Given the high frequency nature of sensor data, this paper assumes the process is measured continuously.

Let~$H_{i,t}^{NX} = H_{i,t}^{N} \otimes H_{i,t}^{X}$ be the $\sigma$-field generated by all past values~$(N_i (s), x_i (s))_{0 \leq s \leq t}$. In this paper, the instantaneous risk of an event at time~$t$ is assumed to depend on the health process, time-in-study, and the event history through a fully parametric conditional hazard function:
\begin{equation}
\label{eq:hazard}
h_i \left( t \Given H_{i,t}^{NX} ; \theta \right) =
\lim_{\delta \to 0} \delta^{-1} \pr \left( N_i(t+\delta) - N_i(t) > 0
  \given H_{i,t}^{NX} \right),
\end{equation}
where~$\theta$ is the parameter vector. For high frequency physiological data, we assume that current risk is log-additive and depends on a linear functional of the health process over some recent window of time and some pre-specified features of the counting process; that is,
\begin{equation}
\label{eq:hazardlinear}
h_i \left( t \given  H_{i,t}^{NX} ; \theta \right) =
h_0 (t; \gamma) \exp \left( g_t \left( H_{i,t}^{N} \right)^{\prime} \alpha
  + \int_{t-\Delta}^t x_i (s) \beta(s) ds  \right)
\end{equation}
where~$h_0(t;\gamma)$ is a parametrized baseline hazard function,~$\Delta$ is an unknown window-length, and $g_t( H_{i,t}^N ) \in \mathbb{R}^p$ is a $p$-length feature vector summarizing the event-history and time-in-study information. The final term~$\int_{t-\Delta}^t x_i(s) \beta(s) ds$ reflects the unknown linear functional form of the impact of the time-varying covariate on current risk.

An alternative to~\eqref{eq:hazardlinear} would be to construct features from the sensor data history~$f_t ( H_{i,t}^{X}) \in \mathbb{R}^q$ and incorporated these features in the place of the final term. Our current approach builds linear features of~$H_{i,t}^X$ directly from the integrated history, avoiding the feature construction problem -- a highly nontrivial issue for high frequency time-series data.  The main caveat is the additional parameter~$\Delta$; however, as long as the estimated~$\hat \Delta$ exceeds~$\Delta$, then resulting estimation is unbiased albeit at a loss of efficiency.  Moreover, sensitivity analysis can be performed to determine how choice of $\hat \Delta$ affects inference.  One limitation of the approach presented here is that only fully parametric hazard models may be fit to the data. However, a spline model for the log baseline hazard affords sufficient model flexibility.

\subsection{Measurement-error models with event processes}
\label{section:memproblems}

One potential criticism for~\eqref{eq:hazard} is that the health process may be measured with error. A common mathematical strategy for joint models is to consider an unobservable, latent process~$\eta_i$ such that~$\bfT_i \indep x_i \given \eta_i$, i.e., the two processes are conditionally independent given the latent trajectory. For example, take~$\eta_i$ to be a zero-mean Gaussian process with
\begin{equation}\label{eq:jm}
x_i(t) = \eta_i (t) + \epsilon_i (t),\quad \text{ and } \log h_i (t
\given \eta ) = \log h_0 (t) + g_t \left( H_{i,t}^N \right)^{\prime}
\alpha + \int_{t-\Delta}^t \eta_i (s) \beta (s) ds
\end{equation}
where~$\epsilon_i(t)$ is a white-noise measurement error term. Thus, $\eta_i (t)$ is the ``true and unobserved value of the longitudinal outcome'' \citep[Sec. 2.1, pp.3]{Rizopoulos2010}. The conditional survival function of the $j$th event~$T_j$ given~$\bfx_i$ and all prior events~$\bfT_{-j} := \{ T_1, \ldots, T_{j-1} \}$ is
\[
\pr \left ( T_j > t+s \Given \H_{i,\tau_i}^X, \bfT_{-j} = \bft_{-j}, T_j > t \right) = E \left( \exp \left( -\int_{t}^{t+s} h_i (u \given \eta ) du \right) \Given \bfx_i, \bfT_{-j} = \bft_{-j} \right),
\]
where the expectation is a Gaussian integral, albeit infinite-dimensional.  For most choices of covariance structure, if the white-noise error term is non-zero then the above calculation will show the conditional survival function depends not only on past $x$-values, but also on future $x$-values, i.e.,~$\bfx$ and~$\bfT$ do not satisfy \emph{independent evolution}~\citep{DempseyPMCC2}. This is quite unnatural, as~\eqref{eq:jm} suggests the instantaneous risk of an event at time~$t$ depends on future values of the sensor process. To ensure independent evolution, we set~$\epsilon (t) \equiv 0$ and treat~$\bfx_i$ as a Gaussian process measured without a white-noise error term.
Our procedure will only rely on the Gaussian assumption when~$\bfx$ is not fully observed (see Section~\ref{section:missingdata}).

\subsection{Likelihood calculation}

For the sake of notational simplicity, we leave the dependency of the conditional hazard function on~$H_{i,t}^{NX}$ implicit, and write~$h_i (t; \theta)$. The component of the log-likelihood related to the event process is then given by
\[
  L_n (\theta) = \sum_{i=1}^{n} \left ( \sum_{j=1}^{k_i}
    \log \left( h_i \left( T_{i,j}; \theta
      \right) \right) - H_{i} \left( \tau_i; \theta \right) \right)
\]
where~$H_{i} (\tau_i ; \theta) = \int_{0}^{\tau_i} h_{i} (t; \theta) dt$ is the cumulative hazard function. Solving the associated score equations~$U_n (\theta) = {\bf 0}$ yields the maximum likelihood estimator~$\hat \theta$, where
\[
U_n (\theta) = \sum_{i=1}^{n} \left ( \sum_{j=1}^{k_i} \frac{h^{(1)}_i
    (T_{i,j}; \theta)}{h_i (T_{i,j}; \theta)} - H^{(1)}_{i} (\tau_i;
  \theta) \right),
\]
with $h_i^{(1)} (T_{i,j}; \theta)$ and $H_i^{(1)} (\tau_{i}; \theta)$ are derivatives with respect to $\theta$.

In classical joint models~\citep{Henderson2000, Tsiatis2004}, time-varying covariates~$x_i (t)$ are observed only intermittently at appointment times. In our current setting, maximizing the likelihood is computationally prohibitive since for any~$\theta$ we must compute the cumulative hazard functions~$H_{i} (\tau_i; \theta)$ which require integration of~$h_i(t;\theta)$ given by~\eqref{eq:hazardlinear} which itself depends on the integral~$\int_{t-\Delta}^t x_i (s) \beta(s) ds$ that is a function of the unknown functional parameter $\beta(\cdot)$.  That is, the risk model now depends on an integrated past history of the time-varying covariate which leads to severe increase in computational complexity.

\subsection{Probabilistic subsampling framework}

To solve the computational challenge we employ a point-process subsampling design to obtain unbiased estimates of the derivative of the cumulative hazards for each subject. The subsampling procedure treats the collected sensor data as a set of \emph{potential observations}. Suppose covariate information is sampled at times drawn from an independent inhomogeneous Poisson point process with known intensity~$\pi_i (t)$. At a subsampled time~$t$, the \emph{windowed covariate history} $\{ x_i (t-s)\}_{0 \leq s \leq \Delta}$ and counting process features~$g_t (H_{i,t}^N)$ are observed. Optimal choice of~$\pi_i (t)$ is beyond the scope of this paper; however, simulation studies have suggested setting the subsampling rate proportional to the hazard function~$h_i (t; \theta)$.

An estimator is design-unbiased if its expectation is equal to that parameter under the probability distribution induced by the sampling design~\citep{Cassel1977}. Let~$\bfD_i \subset [0,t_i]$ denote the random set of subsampled points.  Note, by construction, this random set is distinct from the set of event times with probabilty one, i.e.,~$\pr( \bfT_i \cap \bfD_i = \emptyset) = 1$.  Under subsampling via $\pi_i (t)$, one may compute a Horvitz-Thompson estimator of the derivative of the cumulative hazard $\hat H_{i}^{(1)} (\tau_i; \theta) = \sum_{u \in \bfD_i} h^{(1)}_i (u; \theta)/\pi_i (u)$. An alternative design-unbiased estimator of the derivative of the cumulative hazards is given by
\begin{equation}
\label{eq:WPest}
\hat H_{i}^{(1)} (\tau_i; \theta) = \sum_{u \in (\bfT_i \cup \bfD_i)}
\frac{ h^{(1)}_i ( u; \theta ) }{ \pi_i (u) + h_i ( u ; \theta) }
\end{equation}
Equation~\eqref{eq:WPest} is the estimator suggested by~\cite{Waagepetersen2008}.  This estimator depends on the superposition of the event and subsampling processes. Proposition~\ref{prop:optimal} shows the estimator for~$\theta$ associated with using~\eqref{eq:WPest} is the most efficient within a suitable class of estimators for the derivative of the cumulative hazard function (including the Horvitz-Thompson estimator). Therefore, we restrict our attention to~\eqref{eq:WPest} for the remainder of this paper. Letting
\begin{equation}
\label{eq:waage_weights}
w_i (t; \theta) = \frac{\pi_i (t)}{\pi_i (t) + h_i (t ; \theta)},
\end{equation}
the resulting approximate estimating equations can be re-written as
\begin{equation}
\label{eq:approxscore}
\hat{U}_n (\theta) = \sum_{i=1}^n \left[ \sum_{u \in \bfT_i} w_i(u; \theta)
  \frac{h^{(1)} (u; \theta)}{ h ( u; \theta)}  - \sum_{u \in \bfD_i} w_i(u;
  \theta) \frac{h_i^{(1)} (u; \theta)}{ \pi_i (u) } \right].
\end{equation}
Equation~\eqref{eq:approxscore} represents the approximate score functions built via plug-in of the design-unbiased estimator of the derivative of the cumulative hazard given in~\eqref{eq:WPest}.

\begin{rmk}[Connection to design-based estimation]
The approximate score equations given by~\eqref{eq:approxscore} arise in design-based inference of point processes.  Design-based inference is common for spatial point processes~\citep{Waagepetersen2008} where the spatial varying covariate is observed at a random sample of locations. It is common in mobile health where ecological momentary assessments~\citep{Rathbun2012,Rathbun2016} are used to randomly sample individuals at various time-points to assess their emotional state.  In the current setting, we leverage these ideas to form a subsampling protocol that can substantially reduce computationally complexity.  Therefore, the purpose is quite different. Moreover, the dependence of the intensity function on the recent history of sensor values leads to additional complications that must be addressed.
\end{rmk}

\section{Longitudinal functional principal components within
  event-history analysis}

Probabilistic subsampling converts the single sensor stream~$\bfx_i$ into a sequence of functions observed repeatedly at sampled times~$\bfD_i$ and event times~$\bfT_i$. Such a data structure is commonly referred to as \emph{longitudinal functional data}~\citep{Xiao2013, GoldSmith2015}. Given the large increase in longitudinal functional data in recent years, corresponding analysis has received much recent attention~\citep{Morris2003, MorrisCarroll2006, Baladandayuthapani2008, Di2009, Greven2010, Staicu2010, ChenMuller2012, LiGuan2014}. Here, we combine work by~\cite{Park2018} and~\cite{Goldsmith2011} to construct a computationally efficient penalized functional method for solving the estimation equations~$\hat U_n (\theta)$.

\subsection{Estimation of the windowed covariate history}
\label{sec:margcov}
We start by defining~$X(t,s) = x(t-s)$ to be the sensor measurement~$0 \leq s \leq \Delta$ time units prior to time~$t \in \bfT_i \cup \bfD_i$. We use the sandwich smoother~\citep{Xiao2013} to estimate the mean~$\mu_y(t,s) = \E_y [ X(t,s)]$ where the expectation is indexed by whether $t$ is an event ($y=1$) or subsampled ($y=0$) time respectively. Alternative bivariate smoothers exist, such as the kernel-based local linear smoother~\citep{Hastie2009}, bivariate tensor product splines~\citep{Wood2006}, and the bivariate penalized spline smoother~\citep{MarxEilers2005}. The sandwich smoother was chosen for its computational efficiency and estimation accuracy. We then define~$\tilde X(t,s) = X(t,s) - \hat \mu_y(t,s)$ to be the mean-zero process at each time~$t \in \bfT_i \cup \bfD_i$.

As in~\cite{Park2018}, define the \emph{marginal covariance} by
\[
\Sigma_y (s, s^\prime) = \int_{0}^\tau c_y( (T,s), (T,s^\prime) ) f_y(T) dt.
\]
for~$0 \leq s,s^\prime \leq \Delta$, where~$c_y((t,s), (t,s^\prime))$ is the covariance function of the windowed covariate history~$X(t,\cdot)$ and $f_y(T)$ is the intensity function for event ($y=1$) and subsampled ($y=0$) times respectively. Estimation of~$\Sigma_y$ occurs in two steps. For simplicity, we present the steps for subsampled times (i.e., $y=0$) but the steps are the same for event times as well. First, the pooled sample covariance is calculated at a set of grid points:
\[
\tilde \Sigma_0 (s_r , s_{r^\prime}) = \left( \sum_{i=1}^n \left | \bfD_i \right |\right)^{-1} \left( \sum_{i=1}^n \sum_{t \in \bfD_i} \tilde X (t,s_r) \tilde X (t,s_{r^\prime}) \right).
\]
Due to our concern over independent evolution as discussed in section~\ref{section:memproblems}, we do not assume that each observation is observed with white noise and therefore the diagonal elements of~$\hat \Sigma_0$ are not inflated. Second, the estimator~$\hat \Sigma$ is further smoothed again using the sandwich smoother~\citep{Xiao2013}. Note~\cite{Park2018} smooth the off-diagonal elements, while here we smooth the entire pooled sample covariance matrix. All negative eigenvalues are set to zero to ensure positive semi-definiteness. The result is used as an estimator~$\hat \Sigma_0$ for the pooled covariance $\Sigma_0$.

Next, we take the spectral decomposition of the estimated covariance function; let $\{ \hat \psi^{(0)}_k (s),\hat \lambda^{(0)}_k \}_{k \geq 1}$ be the resulting sequence of eigenfunctions and eigenvalues. The key benefit of the marginal covariance approach is that it allows us to compute a single, time-invariant basis expansion; this reduces the computational burden by avoiding the three dimensional covariance function (i.e., covariance depends on~$t$) and associated spectral decomposition in methods considered by \cite{ChenMuller2012}. Using the Karhunen-Lo{\`e}ve decomposition, we can represent~$X(t,s)$ for $t \in \bfT_i \cap \bfD_i$ by
\[
X(t,s) = \hat \mu_y (t,s) + \sum_{k=1}^{\infty} \hat c^{(y)}_{i,k} (t) \hat
\psi^{(y)}_k (s) \approx \hat \mu_y (t,s) + {\bf c}^{(y)}_{i} (t)^\top \mathbold{\hat
  \psi}^{(y)} (s)
\]
where $\hat c^{(y)}_{i,k} (t) = \int_{t-\Delta}^t \tilde X_i (t,s) \hat \psi^{(y)}_k (s) ds$, ${\bf c}^{(y)}_i (t) = (c^{(y)}_{i,1} (t), \ldots, c^{(y)}_{i,K_x} (t))^\top$, $\mathbold{\hat \psi}^{(y)} (s) = (\hat \psi^{(y)}_1 (s), \ldots, \hat \psi^{(y)}_{K_x} (s))^\top$, and~$K_x < \infty$ is the truncation level of the infinite expansion. Following~\cite{Goldsmith2011}, we set~$K_x$ to satisfy identifiability constraints (see Section~\ref{section:beta} for details). In subsequent sections, we leave the dependence on $y$ (i.e., whether $t \in \bfT_i$ or $\in \bfD_i$) implicit unless required for notational simplicity.

\subsection{Estimation of~$\beta(s)$}
\label{section:beta}

The next step of our method is modeling~$\beta(s)$. Here, we leverage ideas from the penalized spline literature~\citep{Ruppert2003, Wood2006book}. Let~$\mathbold{\phi} (s) = \{ \phi_1 (s), \ldots, \phi_{K_b} (x) \}$ be a spline basis and assume that~$\beta(s) = \sum_{j=1}^{K_b} b_j \phi_{j} (s) = \mathbold{\phi} (t) {\bf b}$ where~${\bf b} = [b_1, \ldots, b_{K_b}]^{\top}$. Thus, the integral in~\eqref{eq:hazardlinear} can be restated as
\begin{align*}
\int_{t-\Delta}^t X(t,s) \beta(s) ds
  &\approx \int_{t-\Delta}^t \left[ \hat \mu(t,s) + {\bf c} (t)^\top
    \mathbold{\hat \psi} (s) \right] \times \left[
    \mathbold{\phi} (s) {\bf b} \right] ds \\
  &= [ M_{i,t}^\top + {\bf c} (t)^\top J_{\hat \psi, \phi} ] {\bf b}
\end{align*}
where~$M_{t} = (M_{1,t}, \ldots, M_{K_b,t})$, $M_{j,t} = \int_{t-\Delta}^t \hat \mu (t,s) \phi_j (s)$, and~$J_{\hat \psi, \phi}$ is a $K_x \times K_b$ dimensional matrix with the~$(k,l)$th entry is equal to~$\int_{0}^\Delta \hat \psi_k (s) \phi_l (s) ds$~\citep{RamsaySilverman2005}.

Given the basis for~$\beta(t)$, the model depends on choice of both~$K_b$ and~$K_x$.  We follow~\cite{Ruppert2002} by choosing $K_b$ large enough to prevent under-smoothing and~$K_x \geq K_b$ to satisfy identifiability constraints. While our theoretical analysis considers truncation levels that depend on~$n$, in practice, we follow the simple rule of thumb and set~$K_b = K_x = 35$. As long as the choices of~$K_x$ and $K_b$ are large enough, their impact on estimation is typically negligible. Below, we will exploit a connection between~\eqref{eq:approxscore} and score equations for a logistic regression model.  Before moving on, we introduce some additional notation. Define
\begin{equation}
\label{eq:approx_hazard}
h_i \left( t \given  H_{i,t}^{NX} ; \theta \right) \approx
\exp \left( Z_{t}^\top \gamma + g_t \left( H_{i,t}^{N} \right)^{\prime} \alpha
  + M_{i,t}^\top {\bf b} + C_{i,t}^\top J_{\hat \psi, \phi} {\bf b} \right)
= \exp \left( W_{i,t}^\top \theta \right),
\end{equation}
where~$\theta = (\gamma, \alpha, {\bf b})$ and~$\exp ( Z_{t}^\top \gamma) =: h_0 (t)$ is the parameterized baseline intensity function. We write~$\tilde U_n (\theta)$ to denote the approximate score function when substituting in~\eqref{eq:approx_hazard} for~\eqref{eq:hazardlinear}.

\subsection{Connection to logistic score functions}
\label{eq:logistication}

We next establish a connection between the above approximate score equations~$\tilde U_n (\theta)$ and the score equations for a logistic regression model. We can then exploit this connection to allow the model to be fit robustly using standard mixed effects software~\citep{Ruppert2002, McCulloch2001}.

\begin{lemma} \normalfont
\label{lemma:logistic}
Under weights~\eqref{eq:waage_weights} and the log-linear intensity function~\eqref{eq:approx_hazard}, the approximate score function~$\tilde U_n (\theta)$ is equivalent to
\[
\sum_{i=1}^n \sum_{t \in \bfT_i \cup \bfD_i } \left[ {\bf 1}[t \in D_i]
  - \frac{1}{1 + \exp \left[- \left( \tilde{W}_{i,t}^\top \theta +
        \log \pi_i (t) \right) \right]} \right] \tilde{W}_{i,t}
\]
where~$\tilde W_{i,t} = -W_{i,t}$. This is the score function for logistic regression with binary response~$Y_i(t)$ for $t \in \bfT_i \cup \bfD_i$ and~$i \in [n]$ where~$Y_i(t) = 1[t \in \bfT_i]$, offset~$\log \pi_i (t)$, and covariates~$\tilde W_{i,t}$.
\end{lemma}

This connection established by Lemma~\ref{lemma:logistic} between our proposed methodology and logistic regression allows us to leverage ``off-the-shelf'' software.  The main complication is pre-processing of the functional data; however, these additional steps can also be taken care of via existing software.  Therefore, the entire data analytic pipeline is easy-to-implement and requires minimal additional effort by the end-user. To see this, we briefly review the proposed inference procedure.

\begin{rmk}[Inference procedure review] \normalfont
Given observed recurrent event and high frequency data~$\{ \bfT_i, \bfx_i \}_{i=1}^n$,
\begin{enumerate}[(a)]
\item For each $i \in [n]$, sample non-event times as a time-inhomogeneous Poisson point process with intensity according to~$\pi_i (t)$
\item \label{p2} Estimate mean~$\mu_y (t,s)$ for $0 \leq s \leq \Delta$ at all event times~$t \in \cup_{i=1}^n \bfT_i$ and sampled non-event times~$t \in \cup_{i=1}^n \bfD_i$.
\item \label{p3} Compute marginal covariance across event times,~$\Sigma_1$, and non-event times,~$\Sigma_0$.
\item \label{p4} Compute eigendecomposition~$\{ \hat \psi_k^{(y)}, \hat \lambda_k^{()} \}$ of marginal covariance~$\Sigma_y$
\item Use the eigendecomposition to construct~$W_{i,t}$ for all $i \in [n]$ and $t \in \bfD_i \cup \bfT_i$
\item \label{point:log} Perform logistic regression with binary outcome~$\{ \Y_i (t) \}$  and offset of~$\log \pi_i (t)$.
\end{enumerate}
\end{rmk}

\noindent Before demonstrating the methodology via simulation in Section~\ref{section:simstudy} and a worked example in Section~\ref{section:example}, we provide a theoretical analysis of our current proposal.

\subsection{Theoretical analysis}

Our theoretical analysis requires assumptions regarding the subsampling procedure, the event process, and the functional data. We state these assumptions and then our main theorems. We start by assuming there exists a~$\tau < \infty$ such that all individuals are no longer at risk (i.e.,~$\tau_i < \tau$ for all~$i$). Moreover, define $R_i (t)$ to be the at-risk indicator for participant~$i$, i.e., $R_i (t) = {\bf 1} [t \in (0,\tau_i)]$. Asymptotic theory provided in Lemma~\ref{lemma:simpleasym} will be proven under regularity conditions A-E in~\cite[pp. 420--421]{Andersen1993} along with the following additional assumptions:

\begin{assumption}[Event process assumptions]
\label{assumption:events}\normalfont
We assume the following holds:
\begin{enumerate}[label=(E.\arabic*)]
\item\label{E1} The subsampling rate is both lower and upper bounded for all at-risk times; that is,~$0 < L < \pi_i (t) < U < \infty$ for all~$i=1,2,\ldots$ and~$t \in [0,\tau]$ such that $R_i (t) = 1$
\item\label{E2} There exists a nonnegative definite matrix~$\Xi (\theta)$ such that
  \[
    n^{-1} \Xi_n (\theta) = n^{-1} \sum_{i=1}^n \int_0^\tau w_i (t; \theta) \times \left[ \frac{h_i^{(1)}(t; \theta) h_i^{(1)} (t;\theta)^\top}{h_i (t; \theta)} \right] \times   R_i (t) dt \overset{P}{\to} \Xi (\theta).
  \]
\item\label{E3} There exists~$M$ such that $|W_{i,j,t}| < M$ for all~$(i,j,t)$.
\item\label{E4} For all~$j,k$
\[
n^{-1} \sum_{i=1}^n \int_0^\tau \left | \frac{d^2}{d\theta_j
    d\theta_k} h_i (t;\theta_0) \right|^2 R_i (t) dt \overset{P}{\to}
C < \infty
\]
as $n \to \infty$.
\end{enumerate}
\end{assumption}

We also require several assumptions due to the truncation of the Karhunen-Lo{\`e}ve decomposition that represents~$X(t,s)$.
\begin{assumption}[Functional assumptions~\citep{Park2018}] \normalfont
\label{assumption:truncation}
The following assumptions are standard in prior work on longitudinal functional data analysis~\citep{Park2018, Yao2005, ChenMuller2012}:
\begin{enumerate}[label=(A.\arabic*)]
\item\label{A1} $X = \{ X(t, s) : (t,s) \in \mathcal{T} \times
  \mathcal{S} \}$ is a square integrable element of the $L^2 (
  \mathcal{T} \times \mathcal{S})$
\item\label{A2} The subsampling and conditional intensity rate functions~$f_y(T)$ are continuous and~$\sup |f_y(T)| < \infty$.
\item\label{A3} $\E[X(t,s) X(t,s^\prime) X(t^\prime,s) X(t^\prime, s^\prime) ] < \infty$ for each $s,s^\prime \in [0,\Delta]$ and~$0 < t, t^\prime < \tau$.
\item\label{A4} $\E[\|X(t,\cdot)\|^4] < \infty$ for each~$0< t < \tau$.
\end{enumerate}
\end{assumption}
Finally, for simplicity, we assume that there exists~$b^\star$ such that~$\beta(t) = \mathbold{\phi} (t) b^\star$; that is, the true function~$\beta(t)$ sits in the span of the spline basis expansion. With the necessary assumptions stated, we can now state Lemma~\ref{lemma:simpleasym}, which provides asymptotic theory.

\begin{lemma} \normalfont
\label{lemma:simpleasym}
Under Assumption~\ref{assumption:events}, Assumption~\ref{assumption:truncation}, and~$\Delta$ known, for large~$n$ the estimator~$\hat \theta_n$ is consistent; moreover,
\[
\sqrt{n} (\hat \theta - \theta) \to N(0, \Xi (\theta)^{-1})
\]
where
\[
  \Xi (\theta) = \int_{0}^{\tau} w(s; \theta) \times \left[ \frac{h^{(1)}(s;
      \theta) \times  h^{(1)} (s;\theta)^{\top}}{h(s; \theta)} \right]
  ds.
\]
and~$\tau$ is the random censoring time of the event process.
\end{lemma}
An estimator for~$\Xi(\theta)$ is
\begin{equation}
\label{eq:fisher}
  \hat \Xi (\theta) = n^{-1} \sum_{i=1}^n \sum_{t \in \bfT_i \cup D_i}
  w_i(t;\theta) (1-w_i(t;\theta)) \left[ \frac{h_i^{(1)}(t;
      \theta)}{h_i (t; \theta)} \right] \times  \left [
    \frac{h_i^{(1)} (t;\theta)}{h_i(t; \theta)} \right]^\top
\end{equation}
For the log-linear intensity model, the sampling-unbiased estimator for~$\hat \Xi(\theta)$ is equivalent to the Fisher information for the previously described logistic regression model.
This implies that subsampling from an inhomogeneous Poisson process, standard logistic regression software can be used to fit the recurrent event model by specifying an offset equal to~$\log \pi_i (t)$. Not only this, Lemma~\ref{prop:optimal} shows weights~\eqref{eq:waage_weights} are optimal within a particular class of weighted estimating equations.

\begin{lemma} \normalfont
\label{prop:optimal}
If the event process is an inhomogeneous Poisson point process with intensity~$h(t; \theta)$ and subsampling occurs via an independent, inhomogeneous Poisson point process with intensity~$\pi (t)$, then~$\hat U_n (\theta)$ are optimal estimating functions (i.e., most efficient) in the class of weighted estimating functions given by~\eqref{eq:approxscore} replacing~\eqref{eq:waage_weights} by any weight function~$w_i (t; \theta)$. This class includes the Horvitz-Thompson estimator under~$w(s; \theta) = 1$.
\end{lemma}

\noindent Proposition~\ref{prop:optimal} ensures the only loss of statistical efficiency is due to subsampling and not using a suboptimal estimation procedure given subsampling.

\subsection{Computation versus statistical efficiency tradeoff}

Under assumption~\ref{assumption:truncation}, we next consider the statistical efficiency of our proposed estimator when compared to complete-data maximum likelihood estimation. While subsampling introduces additional variation, it may significantly reduce the overall computational burden. It is this trade-off that we next make precise. In particular, we consider the following choice of subsampling rate,~$\pi(t) = c \times h(t; \theta)$ for $c>0$. That is, the subsampling rate is proportional to the intensity function with time-independent constant~$c > 0$. Under this subsampling rate, the weight function~\eqref{eq:waage_weights} is equal to $c/ (c+1)$. Under Lemma~\ref{lemma:simpleasym},
\[
\Xi (\theta) = \frac{c}{c+1} \int_0^\tau \frac{ h^{(1)} (t; \theta)
  h^{(1)} (t; \theta)^\top}{h (t; \theta)} dt = \frac{c}{c+1} \Sigma (\theta)
\]
where~$\Sigma(\theta)$ is the Fisher information of the complete-data maximum likelihood estimator.
Therefore the relative efficiency is~$c/(1+c)$. For an upper bound~$H = \max_{t \in (0,\tau)} h(t;\theta)$, if we set~$\pi (t) = c \times H$, then the relative efficiency can be lower bounded by~$c / (c+1)$.

Sensor measurements occur multiple times per second.  Suppose the intensity rate is bounded above by $1$ and the unit time scale was hours. If we then  subsample the data at a rate of~$10$ times per hour, then we have a lower bound on the efficiency of $0.909$. For a 4Hz sensor, this reduces the number of samples per hour from~$4 \times 60 \times 60 = 14,400$ per hour to on average $10$ per hour. While the computational complexity of logistic regression is linear in the number of samples, we get $1440$ times reduction in the data size at the cost of a $0.909$ statistical efficiency. If we sample~$100$ times per hour, then the efficiency loss is only $0.999$, with a $144$ times reduction in data size. Table~\ref{tab:compvseff} provides additional examples for a 4Hz and 32Hz sensor rate respectively.  The data reduction depends on this rate; however, the lower bound on statistical efficiency does not because the subsampling rate only depends on the upper bound of the intensity function. In particular, if the events are rare then subsampling rate can be greatly reduced with no impact to statistical efficiency.

\begin{table}[!th]
\centering
\begin{tabular}{l l r r r r r | c}
\multirow{2}{2.5cm}{Sensor rate} &
\multirow{2}{2.5cm}{Subsampling constant ($c$)}
  & \multicolumn{5}{c}{Upper bound on intensity rate per
    hour}
  & \multirow{2}{2cm}{Statistical efficiency}\\ \cline{3-7}
& & 0.5 & 1 & 3 & 5 & 10 \\ \hline
\multirow{3}{*}{4Hz (EDA)}
& 5 & 5760 & 2880 & 960 & 576 & 288 & 0.833 \\
& 10 & 2880 & 1440 & 480 & 288 & 144 & 0.909 \\
& 100 & 288 & 144 & 48 & 29 & 14 & 0.990 \\ \hline
\multirow{3}{*}{32Hz (ACC)}
& 5 & 46080 & 23040 & 7680 & 4608 & 2304 & 0.833 \\
& 10 & 23040 & 11520 & 3840 & 2304 & 1152 & 0.909 \\
& 100 & 2304 & 1152 & 384 & 230 & 115 & 0.990 \\ \hline
\end{tabular}
\caption{Data reduction (total \# of measurements divided by expected number of subsampled measurements) given sensor rate, subsampling constant and an upper bound on the intensity rate.}
\label{tab:compvseff}
\end{table}

\subsection{Penalized functional regression models}

Recall theoretical results were proven under the assumption that there exists~$b^\star$ such that~$\beta(t) = \phi(t) {\bf b}^\star$. To make this assumption plausible, we set~$K_b$ large enough (but less than~$K_x$ to ensure identifiability) to ensure the spline basis expansion is sufficiently expressive. However, in practice, such a choice of~$K_b$ may lead to overfitting the data. Following~\cite{Goldsmith2011}, we choose the polynomial spline model and set~$\phi(t) = b_0 + b_1 t + b_2 t^2 + \sum_{j=3}^{K_b} (t-\kappa_j)^3$ where~$\{ \kappa_j \}_{j=3}^{K_b}$ are the chosen knots and assume~$\{ b_j \}_{j=3}^{K_b} \sim N(0, \sigma^2 I)$ to induce smoothness on the spline model.  Combining the penalized spline formulation with Lemma~\ref{lemma:logistic} establishes a connection between our approximate score equations and solving a generalized mixed effects logistic regression with offset.

\subsection{Confidence intervals for~$\beta(t)$}

Due to the connection with generalized mixed effects models, we can leverage existing inferential machinery to obtain variance-covariance estimates of model parameters. That is, if~$\hat \Sigma_{bb}$ is the $K_b \times K_b$ dimensional matrix obtained by plugging in the estimates of variance components into the formula for the variance of~$\hat b$, then the standard error for estimate at time~$t_0$ -- i.e., ~$\hat \beta (t_0) = \phi(t_0) \hat {\bf b}$ -- is given by~$\sqrt{ \phi (t_0 ) \hat \Sigma_{bb} \phi(t_0)^\top}$.  Then the approximate 95\% confidence interval can be constructed as~$\hat \beta (t_0) \pm 1.96 \sqrt{\phi (t_0) \hat \Sigma_{bb} \phi(t_0)^\top}$.

We acknowledge two important limitations of confidence intervals obtained via this approach. First, penalization may lead to confidence intervals that perform poorly in regions where~$\hat \beta(t)$ is oversmoothed. Second, we ignore the variability inherent in the longitudinal functional principal component analysis; that is, our estimates ignore the variability in estimation of eigenfunctions~$\hat \psi$ as well as the coefficients~$\hat c_{i,k}(t)$. Joint modeling could be considered as in~\cite{Crainiceanu2010}, however, this is beyond the scope of this article.

\section{Simulation study} \label{section:simstudy}

We next assess the proposed methodology via a simulation study. Here, we assume each individual is observed over five days where each day is defined on the unit interval~$[0,1]$ with $1000$ equally spaced observation times per day. We define~$X(t)$ at the grid of observations as a mean-zero Gaussian process with covariance
\[
  \Sigma (t, t^\prime) =
  \frac{\sigma^2}{\Gamma (v) 2^{\nu-1}}
  \left( \frac{\sqrt{2 \nu} |t-s|}{\rho_t} \right)^{\nu}
  K_v \left( \frac{\sqrt{2 \nu} |t-s|}{\rho_t} \right)
\]
where~$K_v$ is the modified Bessel function of the second kind.  We set~$\nu = 1/2$, $\sigma^2 = 1$, and~$\rho = 0.3$ as well as set~$K_b = K_x = 35$.  For simplicity, we assume~$\Sigma$ is known in computation of the eigendecompositions. Given~$\{ X(t) \}_{0 \leq t \leq 1}$ for a given user-day, we generate event times according a chosen hazard function~$h (t; \theta)$.  To mimic our real data, we set
\[
h(t; \theta) = \exp \left( \theta_0 + \int_{0}^{\Delta} X(t-s) \beta(s) ds \right).
\]
We set~$\Delta$ to mimic a 30-minute window for a $12$-hour day.  We set~$\theta_0 = \log(5/1000)$ to set a baseline risk of approximately 5 events per day. We consider two choices of~$\beta(s)$: (1) $\beta_0 + \exp(- \beta_1 s) $, and (2) $\beta_1  * \sin \left( 2 \pi \frac{s}{\Delta} - \pi/2 \right)$.

We generate $1000$ datasets, each consisting of $500$ user-days.  For a given simulated user-day, we randomly sample non-event times using a Poisson process with rate of a halfhour.  We use the proposed methodology to construct the estimate $\hat \beta_{i,0.5} (t)$ for the $i$th simulated dataset; we then subsample the sampled non-event times with thinning probability $1/2$, $1/4$ and $1/8$.  This results in randomly sampled non-event times given by a Poisson process with rates of an hour, two hours, and four hours.  We can construct the corresponding estimates: $\hat \beta_{i,1} (t)$, $\hat \beta_{i,2} (t)$, and $\hat \beta_{i,4} (t)$ respectively. Subsampling allows us to compare the variance due to subsampling as compared to the variance due to sampling fewer non-event times.

Since we are primarily interested in accuracy, we report the mean integrated squared error (MISE) defined as~$\frac{1}{1000} \sum_{i=1}^{1000} \int_{0}^{\infty} ( \hat \beta_{i,j} (s) - \beta(s))^2 ds$ for $j=0.5,1,2,4$ where~$\beta(s) \equiv 0$ for $s > \Delta$.  The MISE is defined in this manner to account for settings where~$\Delta$ is unknown.  Next, let $\bar \beta_{j} (t) = \frac{1}{1000} \sum_{i=1}^{1000} \hat \beta_{i,j} (t)$ denote the average estimate for $j=0.5,1,2,4$.  Then the squared bias is given by $\int_{0}^{\infty} ( \bar \beta_{j} (s) - \beta (s))^2 ds$ and the variance is given by $\frac{1}{1000} \sum_{j=1}^{1000} \int_{0}^{\infty} ( \hat \beta_{i,j} (s) - \bar \beta_j(s))^2 ds$.  The subsampling variance is defined as $\frac{1}{1000} \sum_{i=1}^{1000} \int_{0}^{\infty} ( \hat \beta_{i,j} (s) - \hat \beta_{i,0.5} (s))^2 ds$.  Table~\ref{tab:mise} show the MISE decomposed into the variance and squared biased as well as the subsampling variance.  To allow fair comparisons across the two choices of $\beta(s)$, all reported numbers are scaled by the integrated square of the true function $\int_{0}^{\infty} \beta(s)^2 ds$. The average runtime (in seconds) is also reported.

Table~\ref{tab:mise} demonstrates that the variance does increase as the sampling rate decreases.  However, the rate of increase in the MISE is quite low relative.  In the second case study, for example, the MISE increases by 6\% while the run time is 3 times faster.  In the first case study, the MISE remains roughly constant but the run time is 4.7 times faster.  This highlights the efficiency-computation trade-off.  Extrapolating to a complete data analysis, we can see that the run time will continue to increase significantly with minimal improvement in the mean integrated squared error.

\begin{table}[!th]
\begin{tabular}{c | p{2cm} p{2cm} p{2cm} p{2cm}}
& \multicolumn{4}{c}{Sampling rate} \\ \cline{2-5}
 & 0.5 & 1 & 2 & 4 \\ \hline
Subsampling variance & - & $9.0 \times 10^{-7}$ & $2.9 \times 10^{-6}$ & $5.8 \times 10^{-6}$  \\
Variance & $7.9 \times 10^{-3}$ & $9.2 \times 10^{-3}$ & $1.0 \times 10^{-2}$ & $1.30 \times 10^{-2}$ \\
Squared Bias & $9.8 \times 10^{-2}$ & $9.5 \times 10^{-2}$ & $9.3 \times 10^{-2}$ & $9.2 \times 10^{-2}$   \\
MISE & $1.1 \times 10^{-1}$ & $1.0 \times 10^{-1}$ & $1.0 \times 10^{-1}$ & $1.0 \times 10^{-1}$ \\ \hline
Avg. runtime (secs) & 317 & 177 & 104 & 68 \\ \hline
Subsampling variance & - & $2.9 \times 10^{-6}$ & $8.9 \times 10^{-6}$ & $2.0 \times 10^{-5}$  \\
Variance & $1.0 \times 10^{-2}$ & $1.3 \times 10^{-2}$ & $2.0 \times 10^{-2}$ & $3.1 \times 10^{-2}$ \\
Squared Bias & $6.9 \times 10^{-2}$ & $6.5 \times 10^{-2}$ & $6.0 \times 10^{-2}$ & $5.4 \times 10^{-2}$   \\
MISE & $7.9 \times 10^{-2}$ & $7.9 \times 10^{-2}$ & $8.0 \times 10^{-2}$ & $8.4 \times 10^{-2}$ \\ \hline
Avg. runtime (secs) & 36 & 23 & 15 & 12 \\ \hline
\end{tabular}
\caption{Mean-integrated squared error, varaince, squared bias, and subsampling variance for $\beta(s)$ given by (1) and (2) respectively.}
\label{tab:mise}
\end{table}

\subsection{Impact of $\Delta$}

A concern with the proposed approach is the selection of window-length, $\Delta$.  Here we investigate the impact of misspecification of the window length for $\beta(t) = \beta_1 \cdot \sin \left( 2 \pi s / \Delta - \pi /2 \right)$ and the true window length ($\Delta^\star$) is set to 32-minutes.  See Appendix~\ref{app:delta} for a similar discussion for $\beta(t) = \beta_0 \exp \left( - \beta_1 s \right)$.  As in the previous simulation, we generate 1000 datasets per condition each consisting of 500 user-days.  For each simulation, we analyze the data using window lengths $\Delta \in (26,29,32,35,37)$.

When the window length is too large, i.e., $\Delta > \Delta^\star$, then asymptotically the estimation is unbiased as $\beta(t) \equiv 0$ for $t > \Delta$; however, we incur a penalty in finite samples, especially for settings where the function is far from zero near $t = \Delta$.  We find the MISE increases the absolute error $|\Delta - \Delta^\star|$ increases, while the variance decreases as a function of $\Delta$.  While the MISE increased for $\Delta < \Delta^\star$, the pointwise estimation error remains low for $t < \Delta$.  This does not hold for $\Delta > \Delta^\star$, where instead we see parameter attenuation, i.e., a bias towards zero in the estimates at each $0 < t< \Delta$.  To capture this, we define a partial MISE as $\frac{\Delta}{\tilde \Delta} \int_0^{\tilde \Delta} ( \hat \beta_{i,j} (s) - \hat \beta (s) )^2 ds$ where $\tilde \Delta = \min(\Delta, \Delta^\star)$, which is the MISE on the subset $0<t< \min (\Delta, \Delta^\star)$ and scaled for comparative purposes.

\begin{table}[!th]
\begin{tabular}{c | c c c c c}
& & \multicolumn{4}{c}{Sampling rate} \\ \cline{3-6}
$\Delta=$ & & 0.5 & 1 & 2 & 4 \\ \hline
\multirow{3}{*}{26} & MISE & 0.400 & 0.404 & 0.409 & 0.413 \\
 & Variance & $1.1 \times 10^{-2}$ & $1.3 \times 10^{-2}$ & $1.5 \times 10^{-2}$ & $1.9 \times 10^{-2}$ \\
  & P-MISE & 0.225 & 0.229 & 0.233 & 0.239 \\ \hline
\multirow{3}{*}{29} & MISE & 0.240 & 0.243 & 0.244 & 0.251 \\
 & Variance & $9.2 \times 10^{-3}$ & $1.1 \times 10^{-2}$ & $1.2 \times 10^{-2}$ & $1.8 \times 10^{-2}$ \\
  & P-MISE & 0.107 & 0.109 & 0.111 & 0.117 \\ \hline
\multirow{3}{*}{32} & MISE & 0.105 & 0.103 & 0.103 & 0.105 \\
 & Variance & $8.0 \times 10^{-3}$ & $8.8 \times 10^{-3}$ & $1.0 \times 10^{-2}$ & $1.4 \times 10^{-2}$ \\
  & P-MISE & 0.105 & 0.103 & 0.103 & 0.105 \\\hline
\multirow{3}{*}{35} & MISE & 0.384 & 0.383 & 0.384 & 0.386 \\
 & Variance & $7.0 \times 10^{-3}$ & $7.7 \times 10^{-3}$ & $9.3 \times 10^{-3}$ & $1.1 \times 10^{-2}$ \\
  & P-MISE & 0.273 & 0.270 & 0.269, & 0.269 \\\hline
\multirow{3}{*}{37} & MISE & 0.589 & 0.588 & 0.589 & 0.591 \\
 & Variance & $6.0 \times 10^{-3}$ & $6.5 \times 10^{-3}$ & $7.9 \times 10^{-3}$ & $1.0 \times 10^{-2}$ \\
  & P-MISE & 0.473 & 0.470 & 0.469 & 0.470 \\\hline
\end{tabular}
\caption{Mean-integrated squared error (MISE), variance, and partial MISE as a function of $\Delta^\star$ and sampling rate when true $\Delta = 32$ minutes.}
\label{tab:mise}
\end{table}

\section{Extensions}

In this section, we demonstrate the flexibility of the proposed approach by exploring extensions in several important directions to ensure these methods are robust for practical use with high frequency data. This section will continue to leverage the connection to generalized functional linear models provided by Lemma~\ref{lemma:logistic}.

\subsection{Multivariate extensions}
\label{section:multiplesensors}

In this section, we extend our model to the case of multiple functional regressors.  That is, suppose $L$ health processes, i.e., $\bfx_i = \{ \bfx_i (t) = (x_{i,1} (t), \ldots, x_{i,L} (t)) \}_{0 < s < \tau_i}$, for each participant is measured at a dense grid of time points. In the suicidal ideation case study, for example, accelerometer is measured at a rate of 32Hz while electrodermal activity (EDA) is measured at a rate of 4Hz.  A multivariate extension of our model~\eqref{eq:hazard} is given by
\begin{equation}
\label{eq:multihazardlinear}
h_i \left( t \given  H_{i,t}^{NX} ; \theta \right) =
h_0 (t; \gamma) \exp \left( g_t \left( H_{i,t}^{N} \right)^{\prime} \alpha
  + \int_{t-\Delta_1}^t x_{i,1} (s) \beta_1(s) ds + \cdots + \int_{t-\Delta_L}^t x_{i,L} (s) \beta_L(s) ds  \right)
\end{equation}
The approach given in Section~\ref{eq:logistication} extends naturally to the multivariate functional setting. For each functional regressor, we estimate the pooled sample covariance $\Sigma_{y,l}$ for $y \in \{0,1\}$ and $l=1,\ldots,L$ as in Section~\ref{sec:margcov}.  Let $\sum_{k=1}^\infty \hat \lambda^{(y)}_{k,l} \hat \psi^{(y)}_{k,l} (s) \hat \psi_{k,l}(t)$ be the spectral decomposition of $\hat \Sigma_{y,l}$.  Then $x_{i,j}(t)$ is approximated using a truncated Karhunen-Loeve $\int_{t-\Delta_l}^t X_{l} (t,s) \beta_l(t) = \left[M_{l,t} + {\bf c}_l (t)^\top J_{\hat \psi_l^{(y)}, \phi_l} \right] {\bf b}_l$.

\subsection{Missing data}
\label{section:missingdata}

Sensor data can often be missing for intervals of time due to sensor wearing issues. In the suicidal ideation case study, for example, there are 2139 self-identified moments of distress across all 91 participants. Of these, 1289 event times had complete data for the prior thirty minutes, 1984 had fraction of missing data on a fine grid less than 30\%, and 1998 had fraction of missing data on a fine grid less than 10\%.

Missing data is a critical issue because~$c_{i,k} (t)$ cannot be estimated if~$X(s,t)$ is not observed for all~$s \in [0,\Delta]$. Moreover, standard errors should reflect the uncertainty in these coefficients when missing data is prevalent. \cite{Goldsmith2011} suggest using best linear unbiased predictors (BLUP) or posterior modes in the mixed effects model to estimate~$c_{i,k} (t)$; however, this is ineffective when there is substantial variability in these estimates.  To deal with this, \cite{Crainiceanu2010} take a full Bayesian analysis. \cite{Yao2005} introduced PACE as an alternative frequentist method. \cite{Petrovich2018} shows that for sparse, irregular longitudinal, the imputation model should not ignore the outcome variable~$Y_i (t)$.

Here we present an extension of~\cite{Petrovich2018} to our setting by leveraging Lemma~\ref{lemma:logistic} and the marginal covariance estimation procedure to construct a multiple imputation procedure. Let~$\bfx_{i} (t)$ denote incomplete sensor data at time~$t$ (i.e., at times~$\{ s_{i,r} \}_{r=1}^{k_{it}}$ in~$[0,\Delta]$. Then
\begin{align}
\label{eq:impute}
\E [ X_i (s,t) \given Y_i(t) = y, \bfx_i (t) ]
  &= \mu_y (s,t) + {\bf a}_{i,t}^\top (s) {\bf B}_{i,t} (\bfx_i (t) -
    \mathbold{\mu}_i (t) ) \\
\text{Cov} \left( X_i (s,t), X_i (s^\prime, t) \given Y_i (t) = y,
  \bfx_i (t) \right)
  &= \Sigma_{t} (s, s^\prime) -
    {\bf a}_{i,t} (s)^\top {\bf B}_{i,t} {\bf a}_{i,t} (s^\prime)
\end{align}
where we have
\[
{\bf a}_{i,t} (s)^\top = \left( \begin{array}{c} \Sigma_t (s_{i,1}, s) \\ \vdots \\
                       \Sigma_t (s_{i,k_{it}}, s) \end{array} \right); \quad
{\bf B}_{i,t}^{-1} = \left(
  \begin{array}{c c c}
    \Sigma_t(s_{i,1}, s_{i,1}) & \Sigma_t(s_{i,1}, s_{i,2}) & \cdots \\
    \Sigma_t(s_{i,1}, s_{i,2}) & \ddots & \vdots \\
    \vdots & \cdots & \Sigma_t (s_{i,k_{it}}, s_{i,k_{it}})
  \end{array} \right),
\]
$\mathbold{\mu}_i (t) = \E [ \bfx_i (t) \given Y_i (t) = y] = \{
\mu_y (s_{j}, t) \}_{j=1}^r$ and~$\mu_y (s,t)$ is the mean of~$X(t,s)$
from group~$y$, and~$\Sigma_t ( s^\prime, s)$ is the covariance
between~$X(s^\prime,t)$ and $X(s,t)$ for $s^\prime, s \in [0, \Delta]$
and~$t \in \mathbb{R}_+$.  To account for overlap between $\Delta$-windows for $t \in \bfT_i \cup \bfD_i$ - i.e., avoid imputing values at a particular time to different values -- imputation is performed sequentially over this set.

\subsubsection{Multiple imputation under uncongeniality}

Multiple imputation yields valid frequentist inferences when the imputation and analysis procedure are congenial~\citep{Meng1994};  the above procedure is derived for function-on-scalar multiple imputation for binary outcomes, which ignored the joint nature of recurrent event analysis in the presence of high frequency sensor data. The main advantage of the above imputation framework is its simplicity and approximate congeniality when events are rare and the sampling rate is low.  The main disadvantage is that the above framework is uncongenial under many events and/or high sampling rates. Again, congeniality ensures good frequentist coverage properties.  A key question is whether we can use the above imputation methods within a general procedure to handle uncongeniality.

To address this, we use the recommendation from~\cite{Bartlett2020} and consider a method that first bootstraps a sample from the dataset and then apply multiple imputation to each bootstrap sample.  This general approach was originally proposed by~\cite{Shao1994} and~\cite{Little2002}. We suppose $B$ bootstraps and $M$ imputations per bootstrap; let $\hat \theta_{b,m}$ denote the estimator for the $m$th imputation of the $b$th bootstrap.  The point estimator is given by $(B)^{-1} \sum_{b=1}^B \hat \theta_{b}$ where $\hat \theta_b = M^{-1} \sum_{m=1}^M \hat \theta_{b,m}$. To construct the confidence interval, we require mean sum of squares with and between boostraps, i.e., $MSW = \frac{1}{B(M-1)} \sum_{b=1}^B \sum_{m=1}^M (\hat \theta_{b,m} - \hat \theta_b ) (\hat \theta_{b,m} - \hat \theta_b )^\top$ and $MSB = \frac{1}{B-1} \sum_{b=1}^B M \cdot (\hat \theta_b - \hat \theta)(\hat \theta_b - \hat \theta)^\top$ respectively. Then the estimator of the variance-covariance matrix of $\hat \theta$ is given by $\hat \Sigma_{B,M} = \left( \frac{B+1}{BM} \right) MSB - MSW/M$.  We obtain the varaiance for~$\beta(t)$ by $\phi (t) \hat \Sigma_{B,M} \phi(t)^\top$.  We follow~\cite{Bartlett2020} and construct confidence intervals based on Satterthwaite’s degrees of freedom, which here is given by
$$
\hat \nu = \phi (t) \hat \Sigma_{B,M} \phi(t)^\top
$$
The bootstrap followed by multiple imputation procedure has been studied extensively by~\cite{Bartlett2020} and is robust to uncongeniality.  The main disadvantage of this approach is its considerable computational intensity.  Recall likelihood calculations were computationally prohibitive by themselves, so combining with bootstrap and MI would further increase this large-scale computation.  The random subsampling framework thus simplifies handling of missing data via connections to function-on-scalar multiple imputation by~\cite{Petrovich2018} as well as to bootstrap to handle uncongeniality by~\cite{Bartlett2020}. Ignoring the computational time of bootstrap sampling, the computational time for the first choise in the simulation study with $B=200$ bootstraps and $M=2$ imputations per bootstrap leads to $35$ hours for a sampling rate of $0.5$ compared to $7$ hours for a sampling rate of $4$, which highlights the benefits of the proposed framework.

\subsection{Multilevel models}

The approach can be extended to multilevel models with functional regressors, which are critical in mobile health where a high degree of individual variation is often observed.  Let $b_i \sim N(0, \sigma_{b}^2)$, then the multilevel extension of~\eqref{eq:hazard} is
\begin{equation}
\label{eq:multilevel}
h_i \left( t \given  H_{i,t}^{NX} ; \theta, b_i \right) \approx
\exp \left( Z_{t}^\top \gamma + g_t \left( H_{i,t}^{N} \right)^{\prime} \alpha
  + \left[ M_{i,t}^\top + C_{i,t}^\top J_{\hat \psi, \phi}\right] (\beta + b_i) \right)
= \exp \left( W_{i,t}^\top \theta  + Z_{i,t}^\top b_i \right),
\end{equation}
where $Z_{i,t} = M_{i,t} + C_{i,t} J_{\hat \psi, \phi}$. Lemma~\ref{lemma:logistic} implies that the random subsampling framework applied to equation~\eqref{eq:multilevel} leads to a penalized logistic mixed-effects model.  As far as the authors are aware, the combination of mixed-effects and $L_2$-penalization on a subset of parameters has not been addressed in existing packages.  To address this, we derive a penalized adaptive Gaussian quadarture (PADQ) for estimation.  For conciseness, the algorithm is presented in Appendix~\ref{app:penlogit}.

\section{A worked example: Adolescent psychiatric inpatient mHealth study} \label{section:example}

During an eight month period in 2018, 91 psychiatric inpatients admitted for suicidal risk to Franciscan Children's Hospital were enrolled in an observational study.  Study data were graciously provided by Evan Kleiman and his study team ({\tt https://kleimanlab.org}). Each study participant wore an Empatica E4~\citep{empaticae3}, a medical-grade wearable device that offers real-time physiological data.  On each user-day, participants were asked to self-identify moments of suicidal distress.  At these times, the participant was asked to press a button on the Empatica E4 device.  The timestamp of the button press was recorded.  One of the primary study goals was to assess the association between sensor-based physiological measurements and self-identified moments of suicidal distress. In particular, the scientific question is whether there are early indicators of escalating distress by monitoring physiological correlates.

A key concern is whether all moments of suicidal distress are recorded.  To ensure this, clinical staff interviewed participants in the evening who were then asked to review their button press activity.  Any events that were identified as incorrect button press activity were removed.  At the end of the 30-day study period, the average number of button presses per day was $2.42$ with a standard deviation of $2.62$.  Investigation of the button press data shows low button-press counts prior to 7AM and a sharp drop off by 11PM. This demonstrates an additional concern: events can only occur when the individual is at-risk, i.e., (A) the individual is wearing the Empatica E4 and (B) is currently awake. To deal with (A) and (B), here we define each study day to begin at 9AM and end at 8PM.

\vspace{0.25cm}
\begin{figure}[!th]
    \centering
    \includegraphics[width=0.8\textwidth]{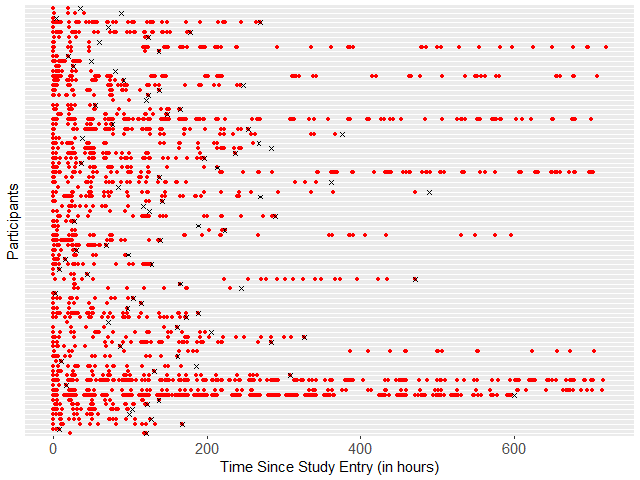}
    \captionof{figure}{User button-presses (red) versus time since study entry (in hours).  The black mark indicates the final sensor measurement time.}
    \label{fig:buttons}
\end{figure}
\vspace{0.25cm}

Figure~\ref{fig:buttons} visualizes button presses versus time since study entry for each user. Day 30 is assumed to censor the observation process. A black mark signals dropout before day 30.   Figure~\ref{fig:buttons} shows the potential heterogeneity in button press rates between users and study days. To assess whether there is between user or between study day variation, Table~\ref{tab:anova} presents a two-way ANOVA decomposition of the button press counts as a function of participant and day in study.  The ANOVA decomposition demonstrates high variation with day in study and across users.

\begin{table}[!th]
\begin{tabular}{ccccccc}\hline
      & DF & Sum Sq & Mean Sq & F value & Pr($>$F) \\ \hline
        Participant & 88 & 2675.4 & 30.4 & 8.9 & $< 2 \times 10^{-16}$ \\
        Day in Study & 29 & 443.3 & 15.3 & 4.5 & $3.9 \times 10^{-13}$ \\
        Residuals & 672 & 2304.2 & 3.4 \\ \hline
      \end{tabular}  \captionof{table}{ANOVA decomposition of daily button press counts}
      \label{tab:anova}
\end{table}

Here, we focus on two physiological processes -- (1) electrodermal activity (EDA), a measure of skin conductance measured at 4Hz, and  (2) \emph{activity index}~\citep{10.1371/journal.pone.0160644}, a coordinate-free feature built from accelerometer data collected at 32Hz that measures physical movement.  Electrodermal activity can be significantly impacted by external factors (e.g., room temperature).  To account for the high between user-day variation, we analyze EDA standardized per study-day and  device. The individual EDA and AI trajectories are highly variable which tends to obscure patterns and trends.  In Figure~\ref{fig:mean_edacc}, the mean trajectories of EDA and AI are plotted in reverse time from button press timestamps, which shows sharp changes in EDA and AI in the 10 minutes prior to button presses.

\begin{figure}[!th]
\centering
\begin{subfigure}{.5\textwidth}
  \centering
  \includegraphics[width=.8\linewidth]{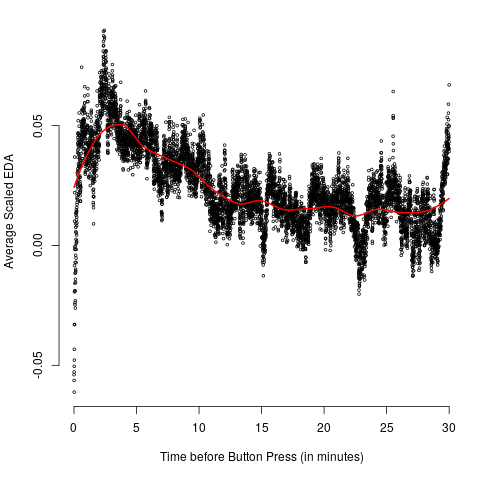}
  \caption{Electrodermal activity (EDA)}
  \label{fig:mean_eda}
\end{subfigure}%
\begin{subfigure}{.5\textwidth}
  \centering
  \includegraphics[width=.8\linewidth]{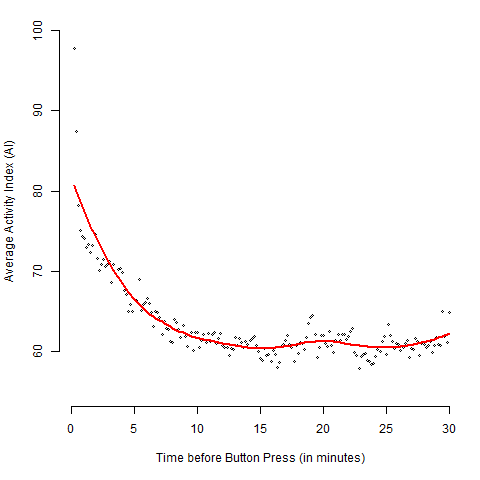}
  \caption{Activity Index (AI)
}  \label{fig:mean_acc}
\end{subfigure}
\caption{Average scaled EDA and AI in the 30 minutes prior to button presses}
\label{fig:mean_edacc}
\end{figure}

\subsection{Complete-case analysis}

Inspection of Figures~\ref{fig:mean_eda} and~\ref{fig:mean_acc} suggest setting $\Delta = 30$ minutes is adequate to capturing the proximal impact of EDA and AI on the risk of a button press.  To ensure minimal loss of efficiency, the subsampling rate was set to once every fifteen minutes.  Given the daily button press rate, this ensures an average of $44$ non-events to $2.5$ events per day.  Based on Table~\ref{tab:compvseff}, this ensures we can achieve a substantial data reduction at a minimal loss of efficiency.  After sampling non-event times, complete-case analyses are performed, i.e., sampled times where sensors included in the model have any level of missing data are ignored.

We analyzed activity index (AI) and electrodermal activity separately as given in equation~~\eqref{eq:hazardlinear} as well as jointly in a multivariate model as in equation~\eqref{eq:multihazardlinear}.  Results did not change substantially.  Figure~\ref{fig:eda_estimate} and~\ref{fig:acc_estimate} presents the estimates from the separate analyses with their associated 95\% confidence intervals.  We highlight in gray the statistically significant regions.  Here, we see that standardized EDA is not associated with increased risk of button press in our population-level model, while activity index sees a positive association in the final few minutes prior to a button press.  In section~\ref{sec:missingdata_heterogeneity}, we investigate whether the population-level results are sensitive to missing data and whether individual-level coefficients,~$\beta(t) + b_i (t)$, effects show stronger dependence on EDA and/or AI for certain individuals.

\begin{figure}[!th]
\centering
\begin{subfigure}{.5\textwidth}
  \centering
  \includegraphics[width=.8\linewidth]{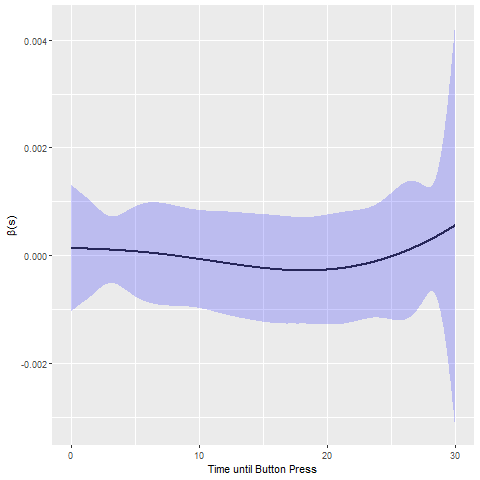}
  \caption{Electrodermal activity (EDA)}
  \label{fig:eda_estimate}
\end{subfigure}%
\begin{subfigure}{.5\textwidth}
  \centering
  \includegraphics[width=.8\linewidth]{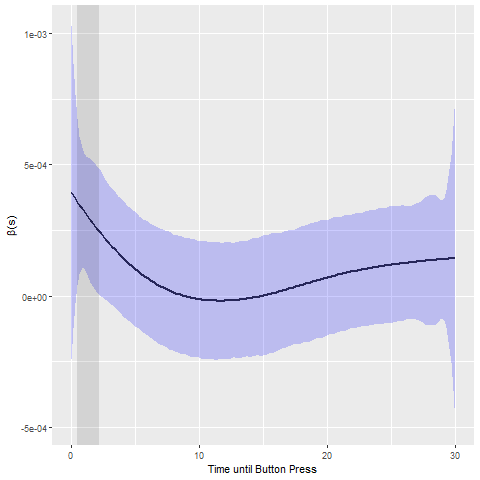}
  \caption{Activity Index (AI)
}  \label{fig:acc_estimate}
\end{subfigure}
\caption{$\beta (t)$ for ACC and EDA with 95\% CI; solid line is complete-case analysis, while dotted line is for }
\label{fig:mean_edacc}
\end{figure}




\section{Discussion}
In this paper, we have presented a methodology for translating a difficult functional analysis with recurrent events problem into a traditional logistic regression.  The translation leveraged subsampling and weighting techniques, specifically the use of weights suggested by~\cite{Waagepetersen2008}, along with flexible functional data analysis methods of~\cite{Goldsmith2011} with marginal covariance methods for longitudinal functional data of~\cite{Park2018}.  The proposed methodology abides by the comment of Robert Gentleman on big data: ``make data as small as possible as fast as possible.''  Subsampling and  weighting converts the problem to well-known territory which allowed us to leverage existing software.  We show limited loss of efficiency when the subsampling is properly tuned to the event rates.  Important extensions to an online sampling algorithm, optimal weighting when the Poisson point process assumption does not hold, and non-linear functional data methods are considered important future work.

\bibliographystyle{plainnat}
\bibliography{si-fda-refs}

\newcommand{\noop}[1]{}
\begin{thebibliography}{42}
\providecommand{\natexlab}[1]{#1}
\providecommand{\url}[1]{\texttt{#1}}
\expandafter\ifx\csname urlstyle\endcsname\relax
  \providecommand{\doi}[1]{doi: #1}\else
  \providecommand{\doi}{doi: \begingroup \urlstyle{rm}\Url}\fi

\bibitem[Andersen et~al.(1993)Andersen, Borgan, Gill, and
  Keiding]{Andersen1993}
P.K. Andersen, O.~Borgan, R.D. Gill, and N.~Keiding.
\newblock Statistical models based on counting processes.
\newblock 1993.

\bibitem[Bai et~al.(2016)Bai, Di, Xiao, Evenson, LaCroix, Crainiceanu, and
  Buchner]{10.1371/journal.pone.0160644}
Jiawei Bai, Chongzhi Di, Luo Xiao, Kelly~R. Evenson, Andrea~Z. LaCroix,
  Ciprian~M. Crainiceanu, and David~M. Buchner.
\newblock An activity index for raw accelerometry data and its comparison with
  other activity metrics.
\newblock \emph{PLOS ONE}, 11\penalty0 (8):\penalty0 1--14, 08 2016.
\newblock \doi{10.1371/journal.pone.0160644}.
\newblock URL \url{https://doi.org/10.1371/journal.pone.0160644}.

\bibitem[Baladandayuthapani et~al.(2008)Baladandayuthapani, Mallick,
  Young~Hong, Lupton, Turner, and Carroll]{Baladandayuthapani2008}
V.~Baladandayuthapani, B.K. Mallick, M.~Young~Hong, J.R. Lupton, N.D. Turner,
  and R.J. Carroll.
\newblock Bayesian hierarchical spatially correlated functional data analysis
  with application to colon carcinogensis.
\newblock \emph{Biometrics}, 64\penalty0 (1):\penalty0 64--73, 2008.

\bibitem[Bartlett and Hughes(2020)]{Bartlett2020}
Jonathan~W Bartlett and Rachael~A Hughes.
\newblock Bootstrap inference for multiple imputation under uncongeniality and
  misspecification.
\newblock \emph{Statistical Methods in Medical Research}, 2020.

\bibitem[Cassel et~al.(1977)Cassel, S\"{a}rndal, and Wretman]{Cassel1977}
C-M. Cassel, S\"{a}rndal, and J.H. Wretman.
\newblock \emph{Foundations of inference in survey sampling}.
\newblock Wiley, New York, 1977.

\bibitem[Chen and M{\"u}ller(2012)]{ChenMuller2012}
K.~Chen and H.-G. M{\"u}ller.
\newblock Modeling repeated functional observations.
\newblock \emph{Journal of the American Statistical Association}, 107\penalty0
  (500):\penalty0 1599--1609, 2012.

\bibitem[Crainiceanu and Goldsmith(2010)]{Crainiceanu2010}
C.~Crainiceanu and J.~Goldsmith.
\newblock Bayesian functional data analysis using winbugs.
\newblock \emph{Journal of Statistical Software}, 32:\penalty0 1--33, 2010.

\bibitem[Dempsey and McCullagh(2019)]{DempseyPMCC2}
W.~Dempsey and P.~McCullagh.
\newblock Vital variables and survival processes.
\newblock Submitted, 2019.

\bibitem[Di et~al.(2009)Di, Crainiceanu, Caffo, and Punjabi]{Di2009}
C.Z. Di, C.M. Crainiceanu, B.S. Caffo, and N.M. Punjabi.
\newblock Multilevel functional principal component analysis.
\newblock \emph{Annals of applied statistics}, 3\penalty0 (1):\penalty0
  458--488, 2009.

\bibitem[Free et~al.(2013)Free, Phillips, Galli, Watson, Felix, Edwards, Patel,
  and Haines]{Freeetal2013}
C.~Free, G.~Phillips, L.~Galli, L.~Watson, L.~Felix, P.~Edwards, V.~Patel, and
  A.~Haines.
\newblock The effectiveness of mobile-health technology-based health behaviour
  change or disease management interventions for health care consumers: A
  systematic review.
\newblock \emph{PLOS Medicine}, 10\penalty0 (1):\penalty0 1--45, 2013.

\bibitem[{Garbarino} et~al.(2014){Garbarino}, {Lai}, {Bender}, {Picard}, and
  {Tognetti}]{empaticae3}
M.~{Garbarino}, M.~{Lai}, D.~{Bender}, R.~W. {Picard}, and S.~{Tognetti}.
\newblock Empatica e3 — a wearable wireless multi-sensor device for real-time
  computerized biofeedback and data acquisition.
\newblock In \emph{2014 4th International Conference on Wireless Mobile
  Communication and Healthcare - Transforming Healthcare Through Innovations in
  Mobile and Wireless Technologies (MOBIHEALTH)}, pages 39--42, 2014.

\bibitem[Goldsmith et~al.(2011)Goldsmith, Bobb, Crainiceanu, Caffo, and
  Reich]{Goldsmith2011}
J.~Goldsmith, J.~Bobb, C.~Crainiceanu, B.~Caffo, and D.~Reich.
\newblock Penalized functional regression.
\newblock \emph{Jounal of Computational and Graphical Statistics}, 20\penalty0
  (4):\penalty0 830--851, 2011.

\bibitem[Goldsmith et~al.(2015)Goldsmith, Zipunnikov, and
  Schrck]{GoldSmith2015}
J.~Goldsmith, V.~Zipunnikov, and J.~Schrck.
\newblock Generalized multilevel functional-on-scalar regression and principal
  component analysis.
\newblock \emph{Biometrics}, 71\penalty0 (2):\penalty0 344--353, 2015.

\bibitem[Greven et~al.(2010)Greven, Crainiceanu, Caffo, and Reich]{Greven2010}
S.~Greven, C.~Crainiceanu, B.~Caffo, and D.~Reich.
\newblock Longitudinal functional principal component analysis.
\newblock \emph{Electronic journal of statistics}, 4:\penalty0 1022--1054,
  2010.

\bibitem[Hastie et~al.(2009)Hastie, Tibshirani, and Friedman]{Hastie2009}
T.~Hastie, R.~Tibshirani, and J~Friedman.
\newblock \emph{The elements of statistical learning}.
\newblock Springer Series in Statistics. Springer, 2009.

\bibitem[Henderson et~al.(2000)Henderson, Diggle, and Dobson]{Henderson2000}
R.~Henderson, P.~Diggle, and A.~Dobson.
\newblock Joint modeling of longitudinal measurements and event time data.
\newblock \emph{Biostatistics}, 1:\penalty0 465--480, 2000.

\bibitem[Klasnja et~al.(2019)Klasnja, Smith, Seewald, Lee, Hall, Luers, Hekler,
  and Murphy]{KlasnjaHS2019}
P.~Klasnja, S.~Smith, N.~Seewald, A.~Lee, K.~Hall, B.~Luers, E.~Hekler, and
  S.A. Murphy.
\newblock Efficacy of contextually tailored suggestions for physical activity:
  A micro-randomized optimization trial of heartsteps.
\newblock \emph{Annals of Behavioral Medicine}, 53:\penalty0 573--582, 2019.

\bibitem[Kleiman et~al.(2018)Kleiman, Turner, Fedor, Beale, Picard, Huffman,
  and Nock]{Kleiman2018}
Evan~M. Kleiman, Brianna~J. Turner, Szymon Fedor, Eleanor~E. Beale, Rosalind~W.
  Picard, Jeff~C. Huffman, and Matthew~K. Nock.
\newblock Digital phenotyping of suicidal thoughts.
\newblock \emph{Depression and Anxiety}, 35\penalty0 (7):\penalty0 601--608,
  2018.

\bibitem[Li and Guan(2014)]{LiGuan2014}
Y.~Li and Y.~Guan.
\newblock Functional principal component analysis of spatio-temporal point
  processes with applications in disease surveillance.
\newblock \emph{Journal of the american statistical association}, 109\penalty0
  (507):\penalty0 1205--1215, 2014.

\bibitem[Little and Rubin(2002)]{Little2002}
R.~J.~A. Little and D.~B. Rubin.
\newblock \emph{Statistical Analysis with Missing Data}.
\newblock Wiley, 2nd edition, 2002.

\bibitem[Marx and Eilers(2006)]{MarxEilers2005}
B.D. Marx and P.H. Eilers.
\newblock Low-rank scale-invariant tensor product smooths for generalized
  additive mixed models.
\newblock \emph{Biometrics}, 62\penalty0 (4):\penalty0 1025--1036, 2006.

\bibitem[McCulloch and Searle(2001)]{McCulloch2001}
Charles~E McCulloch and Shayle~R. Searle.
\newblock \emph{Generalized, Linear and Mixed Models}.
\newblock Wiley, New York, 2001.

\bibitem[Meng(1994)]{Meng1994}
X.L. Meng.
\newblock Multiple-imputation inferences with uncongenial sources of input
  (with discussion).
\newblock \emph{Statistical Science}, 10:\penalty0 538--573, 1994.

\bibitem[Morris and Carroll(2006)]{MorrisCarroll2006}
J.S. Morris and R.J. Carroll.
\newblock Wavelet-based functional mixed models.
\newblock \emph{Journal of the Royal Statistical Society, Series B},
  68\penalty0 (2):\penalty0 179--199, 2006.

\bibitem[Morris et~al.(2003)Morris, Vannucci, Brown, and Carroll]{Morris2003}
J.S. Morris, M.~Vannucci, P.J. Brown, and R.J.~(2003) Carroll.
\newblock Wavelet-based nonparametric modeling of hierarchical functions in
  colon carginogenesis.
\newblock \emph{Journal of the American Statistical Association}, 98\penalty0
  (463):\penalty0 573--583, 2003.

\bibitem[Park and Staicu(2015)]{Park2018}
S.Y. Park and A.M. Staicu.
\newblock Longitudinal functional data analysis.
\newblock \emph{Stat}, 4\penalty0 (1):\penalty0 212--226, 2015.

\bibitem[Petrovich et~al.(2018)Petrovich, Reimherr, and Daymont]{Petrovich2018}
J.~Petrovich, M.~Reimherr, and C.~Daymont.
\newblock Functional regression models with highly irregular designs.
\newblock 2018.

\bibitem[Ramsay and Silverman(2005)]{RamsaySilverman2005}
J.~Ramsay and B.~Silverman.
\newblock \emph{Functional data analysis}.
\newblock Springer, New York, 2005.

\bibitem[Rathbun(2012)]{Rathbun2012}
S.~Rathbun.
\newblock Optimal estimation of poisson intensity with partially observed
  covariates.
\newblock \emph{Biometrika}, 100:\penalty0 277--281, 2012.

\bibitem[Rathbun and Shiffman(2016)]{Rathbun2016}
S.~Rathbun and S.~Shiffman.
\newblock Mixed effects models for recurrent events data with partially
  observed time-varying covariates: Ecological momentary assessment of smoking.
\newblock \emph{Biometrics}, 72:\penalty0 46--55, 2016.

\bibitem[Rizopoulos(2010)]{Rizopoulos2010}
D.~Rizopoulos.
\newblock Jm: An r package for the joint modeling of longitudinal and
  time-to-event data.
\newblock \emph{Journal of Statistical Software}, 35:\penalty0 1--33, 2010.

\bibitem[Ruppert(2002)]{Ruppert2002}
D.~Ruppert.
\newblock Selecting the number of knots for penalized splines.
\newblock \emph{Journal of Computational and Graphical Statistics}, 11\penalty0
  (4):\penalty0 735--757, 2002.

\bibitem[Ruppert et~al.(2003)Ruppert, Wand, and Carroll]{Ruppert2003}
D.~Ruppert, M.~Wand, and R.~Carroll.
\newblock \emph{Semiparametric Regression}.
\newblock Cambridge University Press, Cambridge, 2003.

\bibitem[Shao and Sitter(1996)]{Shao1994}
Jun Shao and Randy~R Sitter.
\newblock Bootstrap for imputed survey data.
\newblock \emph{Journal of the American Statistical Association}, 91\penalty0
  (435):\penalty0 1278--1288, 1996.

\bibitem[Spring(2019)]{Sense2Stop}
B.~Spring.
\newblock Sense2stop: Mobile sensor data to knowledge.
\newblock \url{https://clinicaltrials.gov/ct2/show/NCT03184389}, 2019.

\bibitem[Staicu et~al.(2010)Staicu, Crainiceanu, and Carroll]{Staicu2010}
A.-M. Staicu, C.M. Crainiceanu, and R.J. Carroll.
\newblock Fast methods for spatially correlated multilevel functional data.
\newblock \emph{Biostatistics}, 11\penalty0 (2):\penalty0 177--194, 2010.

\bibitem[Tsiatis and Davidian(2004)]{Tsiatis2004}
A.A. Tsiatis and M.~Davidian.
\newblock Joint modeling of longitudinal and time-to-event data: an overview.
\newblock \emph{Statistica Sinica}, 14:\penalty0 809--834, 2004.

\bibitem[Waagepetersen(2008)]{Waagepetersen2008}
Rasmus Waagepetersen.
\newblock Estimating functions for inhomogeneous spatial point processes with
  incomplete covariate data.
\newblock \emph{Biometrika}, 95\penalty0 (2):\penalty0 351--363, 2008.

\bibitem[Wood(2003)]{Wood2006book}
S.~Wood.
\newblock \emph{Generalized additive models: an introduction with R}.
\newblock Chapman \& Hall, London, 2003.

\bibitem[Wood(2006)]{Wood2006}
S.N. Wood.
\newblock Low-rank scale-invariant tensor product smooths for generalized
  additive mixed models.
\newblock \emph{Biometrics}, 62\penalty0 (4):\penalty0 1025--1036, 2006.

\bibitem[Xiao et~al.(2013)Xiao, Li, and Ruppert]{Xiao2013}
L.~Xiao, Y.~Li, and D.~Ruppert.
\newblock Fast bivariate p-splines: the sandwich smoother.
\newblock \emph{Journal of the Royal Statistical Society: Series B},
  75\penalty0 (3):\penalty0 5770--599, 2013.

\bibitem[Yao et~al.(2005)Yao, M{\"u}ller, and Wang]{Yao2005}
F.~Yao, H.-G. M{\"u}ller, and J.-L Wang.
\newblock Functional data analysis for sparse longitudinal data.
\newblock \emph{Journal of the American Statistical Association}, 100\penalty0
  (470):\penalty0 577--590, 2005.

\end{thebibliography}

\newpage
\appendix

\section{Derivation of the design-unbiased score equations}

\begin{proof}[Proof of Lemma~\ref{lemma:logistic}]
Recall~$\frac{d}{d\theta} \log h_i (t; \theta) = \frac{h^{(1)}_i (t; \theta)}{h_i (t; \theta)}$. Under the log-linear intensity function given by~\eqref{eq:approx_hazard}, $\frac{d}{d\theta} \log h_i (t; \theta) = W_{i,t}$ and
\[
\frac{d}{d \theta} \log \left( \pi_i (t) + h_i (t;\theta) \right) =
\frac{\exp \left( W_{i,t}^\top \theta \right)}{\pi_i (t) + \exp
  \left( W_{i,t}^\top \theta \right)} W_{i,t}
\]
Therefore,
\begin{align*}
\tilde U_n (\theta)
  &= \sum_{i=1}^n \left \{ \sum_{t \in \bfT_i} W_{i,t} - \sum_{t \in \bfT_i \cup \bfD_i} \frac{\exp \left( W_{i,t}^\top \theta \right)}{\pi_i (t) + \exp \left( W_{i,t}^\top \theta \right)}  W_{i,t} \right \} \\
  &= \sum_{i=1}^n \left \{ \sum_{t \in \bfT_i} \frac{\pi_i (t)}{\pi_i (t) + \exp \left( W_{i,t}^\top \theta \right)} W_{i,t} - \sum_{t \in \bfD_i} \frac{\exp \left( W_{i,t}^\top \theta \right)}{\pi_i (t) + \exp \left( W_{i,t}^\top \theta \right)} W_{i,t} \right \} \\
  &= \sum_{i=1}^n \left \{ \sum_{t \in \bfT_i} w_{i} (t; \theta) W_{i,t} -
    \sum_{t \in D_i} (1 - w_i (t; \theta)) W_{i,t} \right \} \\
  &= - \sum_{i=1}^n \sum_{t \in \bfT_i \cup \bfD_i} \left[ {\bf 1} [t \in \bfD_i]  - w_{i} (t; \theta) \right] W_{i,t} \\
  &= \sum_{i=1}^n \sum_{t \in \bfT_i \cup \bfD_i}
    \left[ {\bf 1} [t \in \bfD_i]  - \frac{1}{1 + \exp\left( - (
          \tilde W_{i,t}^\top \theta + \log (\pi_i (t) ) ) \right)}
    \right] \tilde W_{i,t}.
\end{align*}
where~$\tilde W_{i,t} = - W_{i,t}$. This is exactly the score equation
for logistic regression with offset~$\log \pi_i (t)$.
\end{proof}

\begin{proof}[Proof of Lemma~\ref{lemma:simpleasym}]
Define the joint counting process
\[
M_i (t) = N_i (t) - \int_0^t (h_i (s; \theta) + \pi_i (s)) R_i (s) ds.
\]
where~$N_i (t)$ is the counting process with jumps at all~$t \in \bfT_i \cup \bfD_i$. Asymptotic consistency is guaranteed~\cite[Theorem VI.1.1]{Andersen1993} if
\begin{equation}
\label{eq:andersen1}
n^{-1} \tilde U_n (\theta_0 ) \overset{P}{\to} 0,
\end{equation}
and
\begin{equation}
\label{eq:andersen2}
n^{-1} \frac{d}{d \theta^{\top}} \tilde U_n (\theta) \big |_{\theta =
  \theta_0} = \Xi (\theta),
\end{equation}
and
\begin{equation}
\label{eq:andersen3}
\lim_{n \to \infty} P \left( \left | n^{-1} \frac{d^2}{d \theta_j d
      \theta_k} \tilde U_n (\theta) \right | < M \text{ for all }
  j,k \text{ and all } \theta \in \Theta_0 \right) = 1
\end{equation}
To prove~\eqref{eq:andersen1}, we decompose~$\tilde U_n (\theta_0)$ into
three terms
\begin{equation}
\label{eq:decomp}
n^{-1} \tilde U_n (\theta_0) = n^{-1} U_n (\theta_0 )  + n^{-1} \left(
  \hat U_n (\theta_0) - U_n (\theta) \right) + n^{-1} \left(\tilde U_n
  (\theta_0) - \hat U_n (\theta_0) \right)
\end{equation}
where, recall,~$\hat U_n (\theta_0)$ are the logistic score equations with $\int_0^\infty X_i(t,s) \beta(s) ds$ and~$\tilde U_n (\theta_0)$ with the approximation~$C_{i,t}^\top J_{\phi, \psi} {\bf b}$. \cite{Andersen1993} (Theorem VI.1.1) show that~$n^{-1} U_n (\theta_0)
\to 0$. The second term
\begin{align*}
n^{-1} \left( \hat U_n (\theta_0) - U_n (\theta_0) \right) =
  \frac{1}{n} \sum_{i=1}^n \int_0^{\tau} \frac{h^{(1)} (s;
  \theta_0)}{\pi_i (s) + h (s;\theta_0)} d M_i (s).
\end{align*}
where~$h(\cdot; \theta_0)$ is given by~\eqref{eq:hazardlinear}. Lenglart's inequality implies the second term converges in probability to zero under Assumptions~\ref{E1} and~\ref{E2}. The third term satisfies
\begin{align*}
  &n^{-1} \left \| \hat U_n (\theta_0) - \tilde U_n (\theta_0) \right \| \\
  \leq & M n^{-1} \sum_{i=1}^n \sum_{t \in \bfT_i \cup \bfD_i} \bigg \| g
         \left( \tilde W_{i,t}^\top \theta  + \sum_{k=K_x+1}^\infty \hat
         c_{i,k} (t) \int_{t-\Delta}^t \hat \psi_k (s) \phi (s)^\top b
         + \log (\pi_i (t) ) \right) - g \left( \tilde W_{i,t}^\top
         \theta + \log (\pi_i (t)) \right) \bigg \| \\
&\leq M n^{-1} \sum_{i=1}^n \sum_{t \in \bfT_i \cup \bfD_i}
    \sup_{x \in \mathbb{R}} \left \| g^\prime \left( x \right)
  \right\| \times \left \| \sum_{k=K_x+1}^\infty \hat c_{i,k} (t)
  \sum_{l=1}^{K_b} \int_{t-\Delta}^t \hat \psi_k (s) \phi_l (s) b_l
  \right \| \\
&= \frac{M}{4} n^{-1} \sum_{i=1}^n \sum_{t \in \bfT_i \cup \bfD_i}
  \times \sum_{l=1}^{K_b} \left \| \int_{t-\Delta}^t
  \left( X(t,s) - \sum_{k=1}^{K_x} \hat c_{i,k} (t) \psi_k (s) \right)
  \phi_l (s) b_l\right \| \\
&\to \frac{M}{4} \int_0^\tau \sum_{l=1}^{K_b} \left \| \int_{t-\Delta}^t
  \left( X(t,s) - \sum_{k=1}^{K_x} \hat c_{i,k} (t) \psi_k (s) \right)
  \phi_l (s) b_l\right \| (f_0(t) + f_1(t)) dt
\end{align*}
where~$M = \sup \| \tilde W_{i,t} \|$,~$g(x) = 1/(1+\exp(-x))$ is the
expit function.
The second inequality is due to the Taylor remainder theorem. The third is due to~$g^\prime (x) = e^x/(1+e^x)^2 \leq e^0/(1+e^0)^2 = 1/4$, reordering the summation, and rewriting the error term in terms of the difference between~$X(t,s)$ and the approximation with truncation level~$K_x$. Letting~$n \to \infty$ allows us to re-write the outside sums in terms of the integrated difference where the integral is with respect to the joint event and sampling distributions~$(f_0(t) + f_1 (t))$. Finally, under Assumptions~\ref{assumption:truncation}, \cite{Park2018} show
\[
X(t,s) - \sum_{k=1}^{K_x} \hat c_{i,k} (t) \psi_k (s) \overset{P}{\to} 0
\]
as~$K_x \to \infty$, which implies the third term goes to zero as the truncation level goes to infinity. The same argument can be used in conjunction under assumptions~\ref{E1} and~\ref{E4} to prove~\eqref{eq:andersen2}.

To prove~\eqref{eq:andersen3}, we start with equation~\eqref{eq:approxscore}, which we rewrite here for completeness:
\[
\hat{U}_n (\theta) = \sum_{i=1}^n \left[ \sum_{u \in \bfT_i} w_i(u; \theta)
  \frac{h^{(1)} (u; \theta)}{ h ( u; \theta)}  - \sum_{u \in \bfD_i} w_i(u;
  \theta) \frac{h_i^{(1)} (u; \theta)}{ \pi_i (u) } \right].
\]
Let~$N_i^{(e)} (t)$ and $N_i^{(s)} (t)$ be the event and subsampling counting processes with jumps at $t \in \bfT_i$ and $\bfD_i$ respectively.  Then
\begin{equation}
\label{eqapp:upperbound}
\left | n^{-1} \frac{d^2}{d \theta_j d \theta_k} \tilde U_n (\theta)  \right | \leq \frac{1}{n} \sum_{i=1}^n \int_0^\tau w_i (t;\theta) \times H_{in} (t) R_i (t) dN_i^{(e)} (t) +
\frac{1}{n} \sum_{i=1}^n \int_0^\tau \frac{w_i (t;\theta)}{\pi_i (t)} \times G_{in} (t) R_i (t) dN_i^{(e)} (t)
\end{equation}
where $H_{in}(t)$ and $G_{in}(t)$ are from Condition VI.1.1(E) in~\cite[pp. 421]{Andersen1993}. The first term converges in probability to a finite quantity by~\cite{Andersen1993}.  Since~$w_i(t;\theta)/\pi_i (t) = (\pi_i(t)+h_i(t;\theta))^{-1}$, the second term which simplifies to
\[
\frac{1}{n} \sum_{i=1}^n \int_0^\tau \frac{G_{in} (t)}{\pi_i (t)+h_i(t;\theta)} R_i(t) dN_i^{(e)} (t)
\leq
\frac{1}{L \cdot n} \sum_{i=1}^n \int_0^\tau G_{in} (t) R_i (t) dN_i^{(e)} (t)
\]
by Assumption~\ref{E1} (i.e., the subsampling rate is lower bounded by $L$). The right hand side is the optional variation of the local square integrable martingale
$$
\frac{1}{L \cdot n} \sum_{i=1}^n \int_0^t G^{1/2}_{in} (s) dN_i^{(e)} (s)
$$
This martingale has predicable variation process
$$
\frac{1}{L \cdot n} \sum_{i=1}^n \int_0^t G_{in} (s) \pi_i (s) R_i (s) ds
$$
Since both processes have the same limits, conditions VI.1.1(E) in~\cite{Andersen1993} and Assumption~\ref{E1} imply that the second term on the right hand side of~\eqref{eqapp:upperbound} also converges to a finite quantity as $n \to \infty$ for every $K_x$ which completes the proof.

This proves~$\hat \theta_n \overset{P}{\to} \theta_0$ as $n \to \infty$ and $K_x \to \infty$. Note that the argument holds for any choice of weight functions~$w_i (t;\theta)$; implying convergence holds more broadly for all approximate score equations in~\eqref{eq:approxscore}.
\end{proof}

\begin{proof}[Proof of Asymptotic Normality]
Consider a Taylor-expansion of $\tilde U_n (\theta)$ centered at $\theta_0$ when $\hat \theta_n \in \Theta_0$.  We can write as
\begin{align*}
\tilde U_n (\hat \theta_n) &= \tilde U_n (\theta_0) -  (\hat \theta_n - \theta_0) \frac{\partial}{\partial \theta} \tilde U_n (\theta) \mid_{\theta=\theta_0} - (\hat \theta_n - \theta_0) \sum_{l=1}^p (\hat \theta_{ln} - \theta_{l0}) \frac{\partial^2}{\partial \theta \partial \theta_l} \tilde U_n (\theta) \mid_{\theta=\theta^\star} \\
0 &= \tilde U_n (\theta_0) + (\hat \theta_n - \theta_0) \left[ \frac{\partial}{\partial \theta} \tilde U_n (\theta) \mid_{\theta=\theta_0} + \frac{1}{2} \sum_{l=1}^p (\hat \theta_{ln} - \theta_{l0}) \frac{\partial^2}{\partial \theta \partial \theta_l} \tilde U_n (\theta) \mid_{\theta=\theta^\star} \right] \\
n^{1/2} (\hat \theta_n - \theta_0) &= \left[ \frac{1}{n} \frac{\partial}{\partial \theta} \tilde U_n (\theta) \mid_{\theta=\theta_0} + \frac{1}{2n} \sum_{l=1}^p (\hat \theta_{ln} - \theta_{l0}) \frac{\partial^2}{\partial \theta \partial \theta_l} \tilde U_n (\theta) \mid_{\theta=\theta^\star} \right]^{-1} n^{-1/2} \tilde U_n (\theta_0)
\end{align*}
where~$\theta^\star$ is on the line segment between $\hat \theta_n$ and $\theta_0$. First, the term
$$
\frac{1}{n} \frac{\partial}{\partial \theta} \tilde U_n (\theta) \mid_{\theta=\theta_0} + \frac{1}{2n} \sum_{l=1}^p (\hat \theta_{ln} - \theta_{l0}) \frac{\partial^2}{\partial \theta \partial \theta_l} \tilde U_n (\theta) \mid_{\theta=\theta^\star}
\overset{P}{\to} \Xi(\theta_0)
$$
under conditions VI.1.1(A)--(E) in~\cite{Andersen1993} and Assumptions~\ref{E1}--\ref{E4}.  To see this, note the first term converges to $\Xi(\theta_0)$ while the second term is bounded in probability by $p M \| \theta^\star - \theta_0 \|$ for some finite constant~$M$ independent of~$\theta^\star$ where $\| \cdot \|$ is the Euclidean norm.  The second term therefore goes to zero since $\| \theta^\star - \theta_0 \| \to 0$ as $n \to \infty$.

We now consider $n^{-1/2} \tilde U_n (\theta_0 )$. Recall $\tilde U_n (\theta_0)$ can be decomposed into three terms as in~\eqref{eq:decomp}. Here, we assume $K_x \to \infty$ so the final term is negligible.  This leaves two terms
$$
U_n (\theta_0) = \sum_{i=1}^n \int_0^\tau \frac{h^{(1)}(u; \theta)}{h (u;\theta)} dM_i^{(e)} (t)
$$
where~$N_i^{(e)} (t)$ and $N_i^{(s)} (t)$ are the orthogonal event and subsampling counting processes with jumps at times $t \in \bfT_i$ and $t \in \bfD_i$ respectively.
and
\begin{align*}
(\hat U_n (\theta_0) - U_n (\theta_0) )
&= \sum_{i=1}^n \int_0^\tau \frac{h^{(1)}(s; \theta_0)}{\pi_i (s) + h(s; \theta_0)} dM_i (s) \\
&= \sum_{i=1}^n \int_0^\tau \frac{h^{(1)}(s; \theta_0)}{\pi_i (s) + h(s; \theta_0)} dM^{(e)}_i (s) +
\sum_{i=1}^n \int_0^\tau \frac{h^{(1)}(s; \theta_0)}{\pi_i (s) + h(s; \theta_0)} dM_i^{(s)} (s) \\
&= \sum_{i=1}^n \int_0^\tau (1-w(s;\theta_0)) \cdot \frac{h^{(1)}(s; \theta_0)}{h_i (s;\theta_0)} dM_i^{(e)} (s) +
\sum_{i=1}^n \int_0^\tau w(s;\theta_0) \cdot \frac{h^{(1)}(s; \theta_0)}{\pi_i (s)} dM_i^{(s)} (s)
\end{align*}
Taking difference with $U_n(\theta_0)$, the first term becomes
\[
\sum_{i=1}^n \int_0^\tau w(s;\theta_0) \cdot \frac{h^{(1)}(s; \theta_0)}{h_i (s;\theta_0)} dM_i^{(e)} (s)
\]

These are two orthogonal square integrable martingales with quadratic variation given by
\[
\sum_{i=1}^n \int_0^\tau w^2(s;\theta_0) \left[ \frac{h^{(1)}(s; \theta_0)}{h_i (s;\theta_0)} \right] \times \left[ \frac{h^{(1)} (s;\theta_0)}{h_i (s;\theta_0)} \right]^\top h_i (s) ds
\]
and
\[
\sum_{i=1}^n \int_0^\tau w^2(s;\theta_0) \left[ \frac{h^{(1)}(s; \theta_0)}{\pi_i (s)} \right] \times \left[ \frac{h^{(1)} (s;\theta_0)}{\pi_i (s)} \right]^\top \pi_i (s) ds
\]
respectively.  Using the definition $w(s;\theta_0) = \pi_i(s) / (\pi_i(s)+h_i(s;\theta_0))$, we have
\begin{align*}
&\frac{1}{n} \sum_{i=1}^n \int_0^\tau w^2(s;\theta_0) \left[ \frac{h^{(1)}(s; \theta_0)}{h_i (s;\theta_0)} \right] \times \left[ \frac{h^{(1)} (s;\theta_0)}{h_i (s;\theta_0)} \right]^\top h_i (s) ds \\
&+
\frac{1}{n} \sum_{i=1}^n \int_0^\tau w^2(s;\theta_0) \left[ \frac{h^{(1)}(s; \theta_0)}{\pi_i (s)} \right] \times \left[ \frac{h^{(1)} (s;\theta_0)}{\pi_i (s)} \right]^\top \pi_i (s) ds \\
=& \frac{1}{n} \sum_{i=1}^n \int_0^\tau w^2(s;\theta_0) \frac{h^{(1)}(s; \theta_0) [h^{(1)} (s;\theta_0)]^\top}{h_i (s;\theta_0)} \left[ 1 + \frac{h(s;\theta_0)}{\pi_i (s)} \right] ds \\
=& \frac{1}{n} \sum_{i=1}^n \int_0^\tau w^2(s;\theta_0) \frac{h^{(1)}(s; \theta_0) [h^{(1)} (s;\theta_0)]^\top}{h_i (s;\theta_0)} \left[ \frac{\pi_i (s) + h(s;\theta_0)}{\pi_i (s)} \right] ds \\
=& \frac{1}{n} \sum_{i=1}^n \int_0^\tau w(s;\theta_0) \frac{h^{(1)}(s; \theta_0) [h^{(1)} (s;\theta_0)]^\top}{h_i (s;\theta_0)} ds
\to \Xi (\theta_0).
\end{align*}
Returning to the Taylor-expansion, under conditions VI.1.1(A)--(E) in~\cite{Andersen1993} and Assumptions~\ref{E1}--\ref{E4}, we can apply Rebolledo's martingale central limit theorem so that
$$
\sqrt{n} (\hat \theta_n - \theta_0) \overset{D}{\to} N(0, \Xi(\theta_0) ^{-1} \times \Xi(\theta_0) \times \Xi(\theta_0) ^{-1} )
$$
which completes the proof.
\end{proof}

\section{Penalized logistic mixed-effects model}
\label{app:penlogit}

Here we derive a penalized adaptive Gaussian quadrature (PADQ) to fit the   penalized logistic mixed-effects model.  Here, for clarity, we write $Y_{ij} \in \{ 0, 1 \}$ to be the sequence of binary indicators for participant $i = 1,\ldots,n$ and $j$ indexes the event and subsampled times.  We write $W_{ij}$ to denote the time-varying covariate.  Then the score equations for the penalized logistic regression with offset $\log(\pi_{ij})$ is written as
$$
\sum_{i=1}^n \sum_{j=1}^{k_i} \left( y_{ij} - \frac{1}{1+\exp (- W_{ij} \theta - \log (\pi_{ij}) )} \right) W_{ij} - \frac{1}{\sigma^2} \sum_{k=3}^{K_b} \beta_k
$$
Let $B$ be the matrix that extracts the entries in $\theta$ corresponding to $\{ \beta_k\}_{k=3}^{K_b}$ (i.e., $B \theta = ({\bf 0}, \beta_3, \ldots, \beta_{K_b},{\bf 0})$).  Let $\W_i$ be a $k_i \times p$ matrix with the $j$th row equal to $W_{ij}$. Let $\mu(\W_i; \theta, \pi_i)$ be the vector of length $k_i$ with the $j$th entry equal to $\mu(W_{ij}; \theta, \pi_{ij}) = \frac{1}{1+\exp(-W \theta - \log(\pi))}$.  Finally, let $Y_i = (y_{i1},\ldots, y_{i k_i}$ and $E_i (\theta) = Y_i - \mu(\W_i; \theta, \pi_i)$. Then we can succinctly write the above as $\sum_{i=1}^n \W_i^\top E_i (\theta)- \frac{1}{\sigma^2} B \theta$.

We derive the iteratively re-weighted least squares algorithm next. First, take the derivative with respect to $\theta$ to construct the Hessian
$$
- \left[ \sum_{i=1}^n \sum_{j=1}^{k_i} \frac{W_{ij} W_{ij}^\top}{(1+\exp (- W_{ij} \theta - \log (\pi_{ij}) ))^2} \exp(-W_{ij} \theta - \log (\pi_{ij})) + \frac{B}{\sigma^2} \right].
$$
Define the matrix $V_i (\theta)$ be the diagonal matrix with the $(j,j)$th entry equal to $\mu(W_{ij}; \theta, \pi_{ij}) \cdot (1-\mu(W_{ij}; \theta, \pi_{ij})$. Then we can express the Hessian more succinctly as $- \left[ \sum_{i=1}^n \W_i^\top V_i (\theta) \W_i + \frac{B}{\sigma^2} \right]$.  Combining these, starting at an initial estimate, the Newton-Rhapson update given current parameter value $\hat \theta_k$ is given by
$$
\hat \theta_{k+1} = \hat \theta_{k} + \left[ \sum_{i=1}^n \W_i^\top V_i (\hat \theta_k) \W_i + \frac{B}{\sigma^2} \right]^{-1} \left(\sum_{i=1}^n \W_i^\top E_i (\hat \theta_k) - \frac{1}{\sigma^2} B \hat \theta_k \right)
$$

To learn $\sigma^2$, we use the connection to the mixed-effects literature.
Fixing $\theta$, the log-likelihood components related to $\sigma^2$ is given by $-\frac{1}{2\sigma^2} \theta^\top B \theta - \frac{B {\bf 1}}{2} \log (\sigma^2)$.  Differentiating and setting equal to zero yields the estimator
$\hat \sigma^2 = \frac{\theta^\top B \theta}{B {\bf 1}}$.  The complete algorithm is then alternating between Netwon-Rhapson for fixed $\sigma^2$ to estimate $\theta$ and then estimating $\sigma^2$ for fixed $\theta$ using the above formula.

\subsection{Estimation of random-effects via Newton-Rhapson}


In standard logistic mixed-effects model, with $b_i \sim N(0, \Psi)$ and $b_i \in \mathbb{R}^q$, then the likelihood is given by
$$
\begin{aligned}
& \prod_{i=1}^n \int \prod_{j=1}^{k_i} \left( \frac{\exp(W_{ij} \theta + Z_{ij} b_i + \log \pi_{ij})}{1 + \exp(W_{ij} \beta + Z_{ij} b_i + \log \pi_{ij})} \right)^{y_{ij}} \left( \frac{1}{1 + \exp(W_{ij} \beta + Z_{ij} b_i + \log \pi_{ij})} \right)^{1-y_{ij}} \phi (b_i; \sigma_b^2) db_i \\
\propto & \prod_{i=1}^n \int \exp \left[ \sum_{j=1}^{k_i} \left( y_{ij} \left( W_{ij} \theta + Z_{ij} b_i + \log \pi_{ij} \right) - \log \left( 1 + \exp \left( W_{ij} \theta + Z_{ij} b_i + \log \pi_{ij} \right) \right) \right) - b_i^\top \Psi^{-1} b_i /2 \right] db_i \\
\propto& \prod_{i=1}^n \int \exp \left[ g(\theta, b_i, \pi_{ij}, \Psi) \right] db_i,
\end{aligned}
$$
where the second and third line ignore the constant $(2\pi)^{q/2} |\Psi|^{-1/2}$. Differentiating the function $g(\theta, b_i, \pi_{ij}, \Psi)$ with respect to $b_i$, the connection to the prior section can be used.  Re-defining terms appropriately allows us to express the derivatives as
$$
\begin{aligned}
\frac{\partial g(\theta, b_i, \pi_{i}, \Psi)}{\partial b_i} &= Z_i E_i (\theta, b_i) - \Psi^{-1} b_i \\
\frac{\partial g(\theta, b_i, \pi_{i}, \Psi)}{\partial b_i^2} &= - \left[ Z_i V_i (\theta, b_i) Z_i + \Psi^{-1} \right]
\end{aligned}
$$
For fixed $\theta$, the function $g$ is stictly concave function of $b_i$ and therefore we can estimate the unique maximum $\hat b_i$ via Newton-Rhapson method as follows:
$$
\hat b_i^{(k+1)} = \hat b_i^{(k)} + \left[ Z_i V_i (\theta, \hat b_i^{(k)}) Z_i + \Psi^{-1} \right]^{-1}  \left( Z_i E_i (\theta, \hat b_i^{(k)}) - \Psi^{-1} \hat b_i^{(k)} \right)
$$

Fixing $\hat b_i$, we then need to estimate $\Psi$ and $\theta$. Here we employ adaptive gaussian quadrature (cites) as it has been shown that for infrequent events the Laplace approximation tends to underestimate effects (cites).  Let $R_i^\top R_i = Z_i^\top V_i (\theta, b_i) Z_i + \Psi^{-1}$, $z_j$ and $w_j$ for $j=1,\ldots, N_{GQ}$ be the weights for the one-dimensional Gaussian quadrature rule.  Let $\z_\k = (z_{k_1},\ldots, z_{k_q})$, $\tilde b_{i \k} = \hat b_i + R_i^{-1} \z_\k$, and $W_{\k} = \exp ( \| \z_\k \|^2 ) \prod_{l=1}^{q} w_{k_l}$.  Then the AGQ rule is given by
$$
\begin{aligned}
&\int \exp \left[ g(\theta, b_i, \pi_{i}, \Psi) \right] db_i \\
=& \int | R_i |^{-1} \exp \left[ g(\theta, \hat b_i + R_i^{-1} \z, \pi_{i}, \Psi) + \| \z \|^2 / 2 \right] \exp ( - \| \z \|^2 / 2 ) d\z \\
=& (2 \pi)^{q/2} | R_i |^{-1} \sum_{j_1 =1}^{N_{GQ}} \cdots \sum_{j_q =1}^{N_{GQ}} \exp \left[ g(\theta, \tilde b_{i \k}, \pi_{i}, \Psi) \right] W_{\k}
\end{aligned}
$$
The score equations
$$
\sum_{j_1 =1}^{N_{GQ}} \cdots \sum_{j_q =1}^{N_{GQ}}
\left( \sum_{j_1 =1}^{N_{GQ}} \cdots \sum_{j_q =1}^{N_{GQ}} \exp \left[ g(\theta, \tilde b_{i \k}, \pi_{i}, \Psi) \right] W_{\k} \right)^{-1}
\left( \exp \left[ g(\theta, \tilde b_{i \k}, \pi_{i}, \Psi) \right] W_{\k} \right) \frac{\partial}{\partial \theta} g(\theta, \tilde b_{i \k}, \pi_i, \Psi)
$$
Define $\omega(\theta, \tilde b_{i \k}, \pi_i, \Psi)$ to be the weights.  Then we can express score and Hessian as:
$$
\begin{aligned}
&\sum_{\k} \omega(\theta, \tilde b_{i \k}, \pi_i, \Psi) \cdot \frac{\partial}{\partial \theta} g(\theta, \tilde b_{i \k}, \pi_i, \Psi) \\
&\sum_{\k}
\omega(\theta, \tilde b_{i \k}, \pi_i, \Psi) \cdot \frac{\partial^2}{\partial \theta \partial \theta^\top} g(\theta, \tilde b_{i \k}, \pi_i, \Psi) +
\sum_{\k} \omega(\theta, \tilde b_{i \k}, \pi_i, \Psi) \cdot \left [\frac{\partial}{\partial \theta} g(\theta, \tilde b_{i \k}, \pi_i, \Psi) \right] \left [\frac{\partial}{\partial \theta} g(\theta, \tilde b_{i \k}, \pi_i, \Psi) \right]^\top \\
&- \left[ \sum_{\k} \omega(\theta, \tilde b_{i \k}, \pi_i, \Psi) \frac{\partial}{\partial \theta} g(\theta, \tilde b_{i \k}, \pi_i, \Psi) \right] \left[ \sum_{\k} \omega(\theta, \tilde b_{i \k}, \pi_i, \Psi) \frac{\partial}{\partial \theta} g(\theta, \tilde b_{i \k}, \pi_i, \Psi) \right]^\top
\end{aligned}
$$

When $N_{GQ} = 1$, the AGQ approximation is equivalent to a Laplacian approximation.

%
%
%

\section{Additional simulations and case study details}

\subsection{Impact of $\Delta$}
\label{app:delta}

Here we investigate the impact of misspecification of the window length for $\beta(t) = \beta_0 \exp \left( - \beta_1 s \right)$.  To do this, we generate 1000 datasets per condition each consisting of 500 user-days.  For each simulation, we analyze the data using window lengths $\Delta \in (26,29,32,35,37)$.

When the window length is too large, i.e., $\Delta > \Delta^\star$, then asymptotically the estimation is unbiased as $\beta(t) \equiv 0$ for $t > \Delta$; however, we incur a penalty in finite samples, especially for settings where the function is far from zero near $t = \Delta$.  We find the MISE increases and variance slightly decreases as $\Delta$ increases.  While the MISE increased for $\Delta < \Delta^\star$, the pointwise estimation error remains low for $t < \Delta$.  This does not hold for $\Delta > \Delta^\star$, where instead we see parameter attenuation, i.e., a bias towards zero in the estimates at each $0 < t< \Delta$, which is captured by the partial MISE.

\begin{table}[!th]
\begin{tabular}{c | c c c c c}
& & \multicolumn{4}{c}{Sampling rate} \\ \cline{3-6}
$\Delta=$ & & 0.5 & 1 & 2 & 4 \\ \hline
\multirow{3}{*}{26} & MISE & $5.9 \times 10^{-2}$ & $6.1 \times 10^{-2}$ & $6.6 \times 10^{-2}$ & $7.4 \times 10^{-2}$ \\
 & Variance & $1.3 \times 10^{-2}$ & $1.6 \times 10^{-2}$ & $12.4 \times 10^{-2}$ & $3.6 \times 10^{-2}$ \\
  & P-MISE & $7.1 \times 10^{-2}$ & $7.3 \times 10^{-2}$ & $7.8 \times 10^{-2}$ & $8.8 \times 10^{-2}$ \\ \hline
\multirow{3}{*}{29} & MISE & $6.9 \times 10^{-2}$ & $7.0 \times 10^{-2}$ & $7.1 \times 10^{-2}$ & $7.8 \times 10^{-2}$ \\
 & Variance & $1.1 \times 10^{-2}$ & $1.5 \times 10^{-2}$ & $2.0 \times 10^{-2}$ & $3.4 \times 10^{-2}$ \\
  & P-MISE & $7.5 \times 10^{-2}$ & $7.5 \times 10^{-2}$ & $7.7 \times 10^{-2}$ & $8.4 \times 10^{-2}$ \\ \hline
\multirow{3}{*}{32} & MISE & $8.1 \times 10^{-2}$ & $8.0 \times 10^{-2}$ & $8.1 \times 10^{-2}$ & $8.7 \times 10^{-2}$ \\
 & Variance & $1.1 \times 10^{-2}$ & $1.4 \times 10^{-2}$ & $1.9 \times 10^{-2}$ & $3.3 \times 10^{-2}$ \\
  & P-MISE & $8.1 \times 10^{-2}$ & $8.0 \times 10^{-2}$ & $8.1 \times 10^{-2}$ & $8.7 \times 10^{-2}$ \\ \hline
\multirow{3}{*}{35} & MISE & $9.0 \times 10^{-2}$ & $8.8 \times 10^{-2}$ & $8.8 \times 10^{-2}$ & $9.2 \times 10^{-2}$ \\
 & Variance & $1.1 \times 10^{-2}$ & $1.5 \times 10^{-2}$ & $2.2 \times 10^{-2}$ & $3.0 \times 10^{-2}$ \\
  & P-MISE & $8.5 \times 10^{-2}$ & $8.2 \times 10^{-2}$ & $8.1 \times 10^{-2}$ & $8.4 \times 10^{-2}$ \\ \hline
\multirow{3}{*}{37} & MISE & $10.3 \times 10^{-2}$ & $10.1 \times 10^{-2}$ & $9.9 \times 10^{-2}$ & $10.1 \times 10^{-2}$ \\
 & Variance & $1.1 \times 10^{-2}$ & $1.4 \times 10^{-2}$ & $2.0 \times 10^{-2}$ & $3.0 \times 10^{-2}$ \\
  & P-MISE & $9.4 \times 10^{-2}$ & $9.2 \times 10^{-2}$ & $9.0 \times 10^{-2}$ & $9.0 \times 10^{-2}$ \\ \hline
\end{tabular}
\caption{Mean-integrated squared error (MISE), variance, and partial MISE as a function of $\Delta^\star$ and sampling rate when true $\Delta = 32$ minutes.}
\label{tab:appmise}
\end{table}

\end{document}